\documentclass{lmcs} 
\pdfoutput=1

\usepackage{lastpage}
\lmcsdoi{17}{3}{18}
\lmcsheading{}{\pageref{LastPage}}{}{}%
{Sep.~08,~2020}{Aug.~13,~2021}{}


\usepackage{tikz}
\usetikzlibrary{arrows, automata}
\usepackage{thm-restate}

\usepackage{amssymb}

\theoremstyle{plain} 

\def\perf{\mathit{perf}}
\newcommand{\xqed}[1]{%
	\leavevmode\unskip\penalty9999 \hbox{}\nobreak\hfill
	\quad\hbox{\ensuremath{#1}}}

\newcommand{\directions}{\{l,r\}}

\newcommand{\cA}{\mathcal{A}}
\newcommand{\cB}{\mathcal{B}}

\newcommand{\cD}{\mathcal{D}}
\def\nat{\mathbb{N}}
\newcommand{\fintree}{T^{fin}_{\Sigma(\{x_1, \dots, x_n\})}}

\newcommand{\graft}{\mathit{graft}}
\newcommand{\lead}{\alpha_{t_0, \cA, \widehat{STR}}(\phi, v)}

\newcommand{\Power}{\mathcal{P}}

\newcommand{\langUncountable}{\mathfrak{L} [L_0,L^{\neg ba}]}

\def\nat{\mathbb{N}}

\newtheorem{claim}{Claim}[thm]

\newenvironment{claimproof}[1][\color{darkgray}{Proof}]{\begin{proof}[#1]}{\end{proof}}

\newcounter{theoremCounter}
\newcounter{claimCounter}

\newcommand{\saveCounter}{
	\setcounter{theoremCounter}{\value{thm}}
	\setcounter{claimCounter}{\value{claim}}
}
\newcommand{\loadCounter}{
	\setcounter{thm}{\value{theoremCounter}}
	\setcounter{claim}{\value{claimCounter}}
}

\begin{document}

\title[Ambiguity Hierarchy of Regular Infinite Tree Languages]{\texorpdfstring{Ambiguity Hierarchy of\\Regular Infinite Tree Languages}{Ambiguity Hierarchy of Regular Infinite Tree
Languages}}

\author[A. Rabinovich]{Alexander Rabinovich}	
\author[D. Tiferet]{Doron Tiferet}	
\address{Tel Aviv University, Israel}	
\email{rabinoa@tauex.tau.ac.il, sdoron5.t2@gmail.com} 
\urladdr{https://www.cs.tau.ac.il/\textasciitilde{}rabinoa} 





\begin{abstract}
	An automaton is unambiguous if for every input it has at most one accepting computation.
	An automaton is $k$-ambiguous (for $k>0$) if
	for every input it has at most $k$ accepting computations.
	An automaton is boundedly ambiguous if it is $k$-ambiguous for some $k\in \nat$.
	An automaton is finitely (respectively, countably) ambiguous if for every input it has at most finitely (respectively, countably) many accepting computations.

	The degree of ambiguity of a regular language is defined in a natural way.
	A language is $k$-ambiguous (respectively, boundedly, finitely, countably ambiguous) if it is accepted by a $k$-ambiguous
	(respectively, boundedly, finitely, countably ambiguous)
	automaton.
	%
	Over finite words every regular language is accepted by a deterministic automaton.
	Over finite trees every regular language is accepted by an unambiguous automaton.
	Over $\omega$-words every regular language is accepted by an unambiguous B\"uchi automaton
	and by a deterministic parity automaton.
	Over infinite trees Carayol et al. showed that there are ambiguous languages.

	We show that over infinite trees there is a hierarchy of degrees of ambiguity:
	For every $k>1$ there are $k$-ambiguous languages that are not $k-1$ ambiguous; and there are finitely (respectively countably, uncountably) ambiguous languages that are not boundedly (respectively finitely, countably) ambiguous.
\end{abstract}

\maketitle


\section{Introduction}\label{sect:intro}

\paragraph*{Degrees of Ambiguity}
The relationship between deterministic and nondeterministic machines plays a central role
in computer science. An important topic is a comparison of expressiveness, succinctness and complexity of deterministic and nondeterministic models.
Various restricted forms of nondeterminism were suggested and investigated (see \cite{colcombet2015unambiguity,han2017ambiguity} for recent surveys).

Probably, the oldest restricted form of nondeterminism is unambiguity.
An automaton is unambiguous if for every input there is at most one accepting run.
For automata over finite words there is a rich and well-developed theory on the relationship between deterministic, unambiguous and nondeterministic
automata \cite{han2017ambiguity}. All three models have the same expressive power.
Unambiguous automata are exponentially more succinct than deterministic ones, and nondeterministic automata are
exponentially more succinct than unambiguous ones \cite{leiss1981succinct,leung2005descriptional}.

Some problems are easier for unambiguous than for nondeterministic automata. As shown by Stearns and Hunt \cite{stearns1985equivalence}, the equivalence and inclusion problems for unambiguous automata are in polynomial time, while these problems are PSPACE-complete for nondeterministic automata.

The complexity of basic regular operations on languages represented by unambiguous finite automata was investigated in \cite{JJS16}, and tight upper bounds on state complexity of intersection, concatenation and many other operations on languages represented by unambiguous automata were established.

It is well-known that the tight bound on the state complexity of the complementation of nondeterministic automata is $2^n$. In \cite{JJS16}, it was shown that the complement of the language accepted by an $n$-state unambiguous automaton is accepted by an unambiguous automaton with $2^{0.79n+\log n}$ states.

Many other notions of ambiguity were suggested and investigated.
A recent paper \cite{han2017ambiguity} surveys works on the degree of ambiguity
and on various nondeterminism measures for finite automata on words.

An automaton is \emph{$k$-ambiguous} if on every input it has at most $k$ accepting runs; it is \emph{boundedly ambiguous} if it is $k$-ambiguous for some $k$; it is \emph{finitely ambiguous} if on every input it has finitely many accepting runs.

It is clear that an unambiguous automaton is
$k$-ambiguous for every $k>0$, and a $k$-ambiguous automaton is
finitely ambiguous. The reverse implications fail.
For $\epsilon$-free automata over words (and over finite trees), on every input there are at most finitely many accepting runs. 
Hence, every $\epsilon$-free automaton on finite words and on finite trees
is finitely ambiguous.
However, over $\omega$-words
there are nondeterministic automata with uncountably many accepting runs. 
Over $\omega$-words and over infinite trees, finitely ambiguous automata are a proper subclass of the class of countably ambiguous automata, which is a proper subclass of nondeterministic automata.

The cardinality of the set of accepting computations of an automaton over an infinite tree $t$ is bounded by the cardinality of the set of functions from the nodes of $t$ to the state of the automaton,
and therefore, it is at most continuum $2^{\aleph_0}$.
The set of accepting computations on $t$ is definable in Monadic Second-Order Logic (MSO). In B\'{a}r\'{a}ny et al. in \cite{barany2010expressing} it was shown
that the continuum hypothesis holds for MSO-definable families of sets. Therefore, if the set of accepting computations
of an automaton on a tree $t$ is uncountable, then its cardinality is $2^{\aleph_0}$.
Hence, there are exactly two infinite degrees of ambiguity.

The degree of ambiguity of a regular language is defined in a natural way.
	A language is $k$-ambiguous (respectively, boundedly, finitely, countably ambiguous) if it is accepted by a $k$-ambiguous
	(respectively, boundedly, finitely, countably ambiguous)
	automaton.
%
%

Over finite words, every regular language is accepted by a deterministic automaton.
Over finite trees, every regular language is accepted by a deterministic bottom-up tree automaton and by an unambiguous top-down tree automaton.
Over $\omega$-words every regular language is accepted by an unambiguous B\"uchi automaton \cite{omegaLanguagesNonAmbiguous} and by a deterministic parity automaton.

Hence, the regular languages over finite words, over finite trees and over $\omega$-words are unambiguous.

In \cite{carayol2010choice} it was shown that the aforementioned situation is different for infinite trees.
Carayol et al. \cite{carayol2010choice} proved that the language $L_{\exists a} $ of infinite full-binary trees over the alphabet $ \{a,c\} $, defined as $L_{\exists a} := \{t\mid t$ has
at least one node labeled by $a\}$ is ambiguous. The proof is based on the undefinability
of a choice function in Monadic Second-Order logic (MSO) \cite{gurevich1983rabin,carayol2007mso}.

Our results imply that the complement of every countable regular language is not finitely ambiguous.
Since $L_{\exists a} $ is the complement (with respect to the alphabet $\{a,c\}$) of the language that consists of a single tree (i.e. the tree with all nodes labeled by $c$), we conclude that $L_{\exists a} $ is not finitely ambiguous (this strengthens the above mentioned result of \cite{carayol2010choice}).
%
%
Our main result states that over infinite trees there is a hierarchy of degrees of ambiguity:
\begin{thm}[Hierarchy]
  \leavevmode
  \phantomsection \label{th:hierarchy}
	\begin{enumerate}
		\item For every $k>1$ there are $k$-ambiguous languages that are not $(k-1)$-ambiguous.
		\item There are finitely ambiguous languages that are not boundedly ambiguous.
		\item There are countably ambiguous languages that are not finitely ambiguous.
		\item There are uncountably ambiguous languages that are not countably ambiguous.
	\end{enumerate}
\end{thm}
 Some natural tree languages that witness items  (1), (3) and (4) of Theorem \ref{th:hierarchy} are described in the   examples below.
We have not found  a ``natural'' finitely ambiguous language which is  not boundedly ambiguous (Theorem \ref{th:hierarchy}(2)). 
\begin{exas}
 Let $T_\Sigma^\omega$ be the set of all infinite full-binary trees over an alphabet $\Sigma$.
		Let $\Sigma_k = \{c, a_1, a_2, ..., a_k\}$, and let 
 $L_{\neg a_i} := \{t \in T_{\Sigma_k}^\omega \mid$ no node in $t$ is labeled by $a_i \}$
     for $1 \leq i \leq n$. Define:
		\begin{enumerate}
\item $L_{\neg a_1 \vee \dots \vee \neg a_k} := L_{\neg a_1} \cup \dots \cup L_{\neg a_k}$. We show that this language  is $k$-ambiguous, but is not $(k-1)$-ambiguous (see  Section~\ref{sect:k-ambig}).  In \cite{DBLP:conf/csl/BilkowskiS13} it was shown that $L_{\neg a_1 \vee   \neg a_2}$ is two ambiguous.
    \item $L_{\exists a_1} := \{t \in T_{\Sigma_1}^\omega \mid$ there exists an $a_1$-labeled node in $t \}$. This is a countably ambiguous language that is not finitely ambiguous (see Section~\ref{sect:not-finite}).
\item $L_{no-\max-a_1} := \{t \in T_{\Sigma_1}^\omega \mid$ above every $a_1$-labeled node in $t$ there is an $a_1$-labeled node$ \}$. This is an uncountably ambiguous language that is not countably ambiguous (see Section~\ref{sect:uncount}).
\end{enumerate}
\end{exas}

\paragraph{Organization of the Paper:}
In Section~\ref{sect:prel} we recall notations
and basic results about automata and monadic second-order logic.
In Section~\ref{sect:simple} simple properties of languages are proved.
Section~\ref{sect:not-finite} gives a sufficient condition for a language to be not finitely ambiguous. The proof techniques used in Section~\ref{sect:not-finite} refine the proof techniques of \cite{carayol2010choice}~-- we rely on the fact that a choice function is not MSO-definable to obtain a lower bound for degree of ambiguity that is higher than the bound which was presented in \cite{carayol2010choice}.
Section~\ref{sect:k-ambig} deals with $k$-ambiguous languages~-- for every $k\in \nat$, we describe a $k$-ambiguous language that is not $(k-1)$-ambiguous.
Section~\ref{sect:finite} provides an example of a finitely ambiguous language which is not boundedly ambiguous. Section~\ref{sect:uncount} introduces a scheme for obtaining uncountably ambiguous languages from languages that are not boundedly ambiguous, and presents some natural examples of uncountably ambiguous languages.
In Section~\ref{sect:countable-langauges-ambiguity}, relying on the characterization of countable regular languages given by Niwi\'nski \cite{Damian}, we prove that every countable tree language is unambiguous.
The conclusion is given in Section~\ref{sect:concl}.

%


An extended abstract of this paper was published in \cite{rabinovich2020ambiguityHierarchy}.
Here we added the proofs that were sketched or missing in \cite{rabinovich2020ambiguityHierarchy}, presented natural examples of uncountably ambiguous languages (in Section~\ref{sect:uncount}), and added Section~\ref{sect:countable-langauges-ambiguity} in which we prove that countable languages are unambiguous using Niwi\'nski's Representation for Countable Languages.

\section{Preliminary} \label{sect:prel}
We recall here standard terminology and notations about trees, automata and logic \cite{PinPerrin, Rabin}.
In Subsection~\ref{subsect:mso} we also recall Gurevich-Shelah's theorem about undefinability of choice function and derive Lemma~\ref{lemma-definable-selection-functions} which plays
an important role in our proofs.

\subsection{Trees}
We view the set $\{l,r\}^*$ of finite words over alphabet $\{l,r\}$ as the domain of a full-binary tree, where the empty word $\epsilon$ is the root of the tree, and for each node $v \in \{l,r\}^*$, we call $v \cdot l$ the left child of $v$, and $v \cdot r$ the right child of $v$.

We define a tree order ``$\leq$'' as a partial order such that $\forall u,v \in \{l,r\}^*: u \leq v$ iff $u$ is a prefix of $v$.
Nodes $u$ and $v$ are incomparable~-- denoted by $u\perp v$~-- if neither $u\leq v$ nor $v\leq u$; a set $U$ of nodes is an \textbf{antichain}, if its
elements are incomparable with each other.

We say that an infinite sequence $\pi = v_0, v_1, \dots$ is a \textbf{tree branch} if $v_0 = \epsilon$ and $\forall i \in \nat: v_{i+1} = v_i \cdot l$ or $v_{i+1} = v_i \cdot r$.



If $\Sigma$ is a finite alphabet, then a $\Sigma$-labeled full-binary tree $t$ is a labeling function $t: \{l,r\}^* \rightarrow \Sigma$.
We denote by $T^\omega_\Sigma$ the set of all $\Sigma$-labeled full-binary trees. We often use ``tree'' for ``labeled full-binary tree.''

Given a $\Sigma$-labeled tree $t$ and a node $v \in \{l,r\}^*$, the tree $t_{\geq v}$ (called the subtree of $t$, rooted at $v$) is defined by $t_{\geq v}(u) := t(v \cdot u)$ for each $u \in \{l,r\}^*$.

\paragraph{Grafting}
Given two labeled trees $t_1$ and $t_2$ and a node $v \in \{l,r\}^*$, the grafting of $t_2$ on $v$ in $t_1$, denoted by $t_1 [t_2/v]$, is the tree $t$ that is obtained from $t_1$ by replacing the subtree of $t_1$ rooted at $v$ by $t_2$. Formally,
$t(u) := \begin{cases}
t_2(w) & \exists w \in \{l,r\}^*: u = v \cdot w\\
t_1(u) & \text{otherwise}
\end{cases}$

More generally, given a tree $t_1$, an antichain $Y \subseteq \{l,r\}^*$ and a tree $t_2$, the grafting of $t_2$ on $Y$ in $t_1$, denoted by $t_1 [t_2/Y]$, is obtained by replacing each subtree of $t_1$ rooted at a node $y \in Y$ by the tree $t_2$.


\paragraph{Tree Language}
A language $L$ over an alphabet $\Sigma$ is a set of $\Sigma$-labeled trees. We denote by $\overline{L}:=T^\omega_\Sigma \setminus L$ the complement of $L$.


\subsection{Automata} We define the following notions of automata:


 \paragraph{Parity $\omega$-word Automata (PWA)}
 A PWA is a tuple $(Q_\cA, \Sigma, Q_I, \delta, \mathbb{C})$ where $\Sigma$ is a finite alphabet, $Q$ is a finite set of states, $Q_I \subseteq Q$ is a set of initial states, $\delta \subseteq Q \times \Sigma \times Q$ is a transition relation, and $\mathbb{C}: Q \rightarrow \nat$ is a coloring function. A run of $\cA$ on an $\omega$-word $y = a_0 a_1 \dots$ is an infinite sequence $\rho = q_0 q_1 \dots$ such that $q_0 \in Q_I$, and $(q_i, a_i, q_{i+1}) \in \delta$ for all $i \in \nat$. We say that $\rho$ is accepting if the maximal number that occurs infinitely often in $\mathbb{C}(q_0) \mathbb{C}(q_1) \dots$ is even.

 \paragraph{Language}
 We denote the set of all accepting runs of $\cA$ on $y$ by $ACC(\cA, y)$. The language of $\cA$ is defined as $L(\cA) := \{ y \in \Sigma^\omega \mid ACC(\cA, y) \neq \emptyset \}$.


\paragraph{Parity Tree Automata (PTA) on Infinite Trees}
A PTA is a tuple $(Q_\cA, \Sigma, Q_I, \delta, \mathbb{C})$ where $\delta \subseteq Q \times \Sigma \times Q \times Q$, and $\Sigma$, $Q$, $Q_I$, $F$ are defined as in PWA. A computation of $\cA$ on a tree $t$ is a function $\phi: \{l,r\}^* \rightarrow Q$ such that $\phi(\epsilon) \in Q_I$, and $\forall v \in \{l,r\}^*: (\phi(v), t(v), \phi(v \cdot l), \phi(v \cdot r)) \in \delta$. We say that $\phi$ is accepting if for each tree branch $\pi = v_0 v_1 \dots$, the maximal number that occurs infinitely often in $\mathbb{C}(\phi(v_0)) \mathbb{C}(\phi(v_1)) \dots$ is even.

Given a PTA $\cA = (Q_\cA,\Sigma, Q_I, \delta_\cA, \mathbb{C}_\cA)$ and a set $Q' \subseteq Q_\cA$, we define $\cA_{Q'} := (Q_\cA,\Sigma, Q', \allowbreak \delta_\cA, \mathbb{C}_\cA)$ as the automaton obtained from $\cA$ by replacing the set of initial states $Q_I$ with $Q'$. For a singleton $Q' = \{q\}$, we simplify this notation by $\cA_q := \cA_{Q'}$.

\paragraph{Language}
We denote the set of all accepting computations of $\cA$ on $t$ by $ACC(\cA, t)$. The language of $\cA$ is defined as $L(\cA) := \{ t \in T^\omega_\Sigma \mid ACC(\cA, t) \neq \emptyset \}$. A tree language is said to be \emph{regular} if it is accepted by a PTA.

A state $q \in Q$ of a PTA $\cA$ is called useful if there is a tree $t \in L(\cA)$, a computation $\phi \in ACC(\cA, t)$ and a node $v \in \{l,r\}^*$ such that $\phi(v) = q$. Throughout the paper we will assume that all states of PTA are useful.

\paragraph{Degree of Ambiguity of an Automaton}
We denote by $|X|$ the cardinality of a set $X$.
An automaton $\cA$ is $k$-ambiguous if $|ACC(\cA, t)| \leq k$ for all $t \in L(\cA)$;
$\cA$ is unambiguous if it is $1$-ambiguous;
$\cA$ is boundedly ambiguous if there is $k \in \nat$ such that $\cA$ is $k$-ambiguous;
$\cA$ is finitely ambiguous if $ACC(\cA, t)$ is finite for all $t$;
$\cA$ is countably ambiguous if $ACC(\cA, t)$ is countable for all $t$.

The degree of ambiguity of $\cA$ (notation $da(\cA)$) is defined by $da(\cA) := k$ if $\cA$ is $k$-ambiguous and either $k = 1$ or $\cA$ is not $k-1$ ambiguous,
$da(\cA) := finite$ if $\cA$ is finitely ambiguous and not boundedly ambiguous, $da(\cA) := \aleph_0$ if $\cA$ is countably ambiguous and not finitely ambiguous, and $da(\cA) := 2^{\aleph_0}$ if $\cA$ is not countably ambiguous.

We order the degrees of ambiguity in a natural way: $i < j < finite < \aleph_0 < 2^{\aleph_0}$, for $i< j \in \nat$.

\subparagraph{\textbf{Degree of  Ambiguity of a Language}} We say that a regular tree language $L$ is unambiguous (respectively, $k$-ambiguous, finitely ambiguous, countably ambiguous) if it is accepted by an unambiguous (respectively, $k$-ambiguous, finitely ambiguous, countably ambiguous) automaton.
We define $da(L) := min_\cA \{da(\cA) \mid L(\cA) = L \}$.

\subsection{Monadic Second-Order Logic}\label{subsect:mso}
We use standard notations and terminology about monadic second-order logic (MSO) \cite{Rabin,trakhtenbrotfinite,thomas1990automata}.

Let $\tau$ be a relational signature. A structure (for $\tau$) is a tuple $M=(D, \{R^M \mid R\in \tau\})$ where $D$ is a domain, and each symbol $R \in \tau$ is interpreted as a relation $R^M$ on $D$.
%
%

MSO-formulas use first-order variables, which are interpreted by elements of the structure, and monadic second-order variables, which are interpreted as sets of elements. Atomic MSO-formulas are of the following form:
\begin{itemize}
	\item $R(x_1, \dots, x_{n})$ for an $n$-ary relational symbol $R$ and first order variables $x_1, \dots, x_n$
	\item $x = y$ for two first-order variables $x$ and $y$
	\item $x \in X$ for a first-order variable $x$ and a second-order variable $X$
\end{itemize}
MSO-formulas are constructed from the atomic formulas, using boolean connectives, the first-order quantifiers, and the second-order quantifiers.

We write $\psi(X_1, \dots, X_n, x_1, \dots, x_m)$ to indicate that the free variables of the formula $\psi$ are $X_1, \dots, X_n$ (second order variables) and $x_1, \dots, x_m$ (first order variables).
We write $ M \models\psi(A_1,\dots, A_n,a_1,\dots a_m)$ if $\psi$ holds in $M$ when subsets $A_i$ are assigned to $X_i$ for $i=1,\dots,n$ and
elements $a_i$ are assigned to variables $x_1,\dots,x_m$ for $i=1,\dots ,m$.

\paragraph{Coding}
Let $\Delta$ be a finite set. We can code a function from a set $D$ to $\Delta$ by a tuple of unary predicates on $D$.
This type of coding is standard, and we shall use explicit variables that range over such mappings and expressions of the form ``$F(u)=d$''
(for $d\in \Delta$) in MSO-formulas, rather than their codings.

Formally, for each finite set $\Delta$ we have second-order variables $X_1^\Delta, X_2^\Delta, \dots$
that range over the functions from $D$ to $\Delta$, and atomic formulas
$X_i^\Delta(u)=d $ for $d\in \Delta$ and $u$ a first order variables \cite{trakhtenbrotfinite}.
Often the type of the second order variables will be clear from the context and we drop the superscript $\Delta$.

\paragraph{Definable Relations}
The powerset of $D$ is denoted by $\Power (D)$. 
We say that a relation $R \subseteq \Power(D)^n \times D^m$ is MSO-definable in a structure $S$ with universe $D$ if there is an MSO-formula $\psi(X_1, \dots, X_n, x_1, \dots, x_m)$ such that $R = \{(D_1, \dots, D_n, u_1, \dots, u_m) \in \Power(D)^n \times D^m \mid S \models \psi(D_1 \dots, D_n, u_1 \dots, u_m)\}$.

An element $d \in D$ is MSO-definable in a structure $S$ if there is a formula $\psi(x)$ such that $S \models \phi(u)$ iff $u = d$. A set $U \subseteq D$ is MSO-definable if there is a formula $\phi(X)$ such that $S \models \phi(V)$ iff $V = U$.
%
%
A function is MSO-definable if its graph is. 

The unlabeled binary tree is the structure $(\{l,r\}^*, \{E_l, E_r\})$ where $E_l$ and $E_r$ are binary symbols, respectively interpreted as $\{(v, v \cdot l) \mid v \in \{l,r\}^*)\}$ and $\{(v, v \cdot r) \mid v \in \{l,r\}^*)\}$.

It is easy to verify the correctness of the following lemma:
\begin{lem}
	The following relations are MSO-definable in the unlabeled full-binary tree.
	\begin{itemize}
		\item The ancestor relation $\leq$.
		\item ``A set of nodes is a branch,'' ``A set of nodes is an antichain.'' 
		\item Let $\cA = (Q,\Sigma, Q_I, \delta, \mathbb{C})$ be a PTA. We use $\phi$ for a function $\{l,r\}^* \rightarrow Q$ and $\sigma$ for a function $\{l,r\}^* \rightarrow \Sigma$.
		\begin{itemize}
			\item ``$\phi$ is a computation of $\cA$ on the tree $\sigma$.'' 
			\item ``$\phi$ is an \textbf{accepting} computation of $\cA$ on the tree $\sigma$.'' 
		\end{itemize}
	\end{itemize}
\end{lem}
The following two fundamental theorems were  proved by Rabin in his famous 1969 paper \cite{Rabin}.
\begin{thm}[Rabin \cite{Rabin}]
	A tree language is regular iff it is MSO-definable in the unlabeled binary tree structure.
\end{thm}
%
A labeled tree is regular iff it has finitely many different subtrees. An equivalent definition is:
a tree is regular iff its labeling is MSO-definable \cite{Rabin}. Hence, for every regular $\Sigma$-labeled tree $t_0$ there is an MSO-formula $\psi_{t_0}(\sigma^\Sigma$), where $\sigma^\Sigma$ is the coding of $\{l,r\}^* \rightarrow \Sigma$, that is satisfied by $t: \{l,r\}^* \rightarrow \Sigma$ iff $t=t_0$.
\begin{thm}[Rabin's basis theorem \cite{Rabin}]
	Any non-empty regular tree language contains a regular tree. 
\end{thm}
\paragraph{Choice Function}
A choice function is a mapping that assigns to each non-empty set of nodes one element from the set.

\begin{thm}[Gurevich and Shelah \cite{gurevich1983rabin}] \label{theorem-MSO-choice-function}
	There is no MSO-definable choice function on the full-binary tree.
\end{thm}
The following lemma follows from Theorem \ref{theorem-MSO-choice-function}. It  plays a key role in our  proofs in Section~\ref{sect:not-finite},
where sufficient conditions are provided  for a language to be
 not finitely ambiguous.

\begin{lem} \label{lemma-definable-selection-functions}
	There is no MSO-definable function that assigns to every non-empty antichain $Y$ a finite non-empty subset $X\subseteq Y$.
\end{lem}
\begin{proof}
	Assume, for the sake of contradiction, that a function that returns a finite non-empty  subset for each non-empty antichain is MSO-definable in the unlabeled full-binary tree, by an MSO-formula $FiniteAntichainSubset(X,Y)$.
	
	\begin{claim}[Choice function over finite sets] \label{claim-finite-choice}
		There is an MSO-definable function that assigns to each non-empty finite set $X \subseteq \{l,r\}^*$ an element $x \in X$.
	\end{claim}
	\begin{claimproof}
          We first define a lexicographic order ``$\leq_{\mathit{lex}}$'' on $\{l,r\}^*$, by $u \leq_{\mathit{lex}} v$ iff $u$ is a prefix of $v$ or $u = w \cdot l \cdot u'$ and $v = w \cdot r \cdot v'$ for some $w, u', v' \in \{l,r\}^*$.
		
                It is easy to verify that $\leq_{\mathit{lex}}$ is MSO-definable in the unlabeled full-binary tree. $\leq_{\mathit{lex}}$ is a linear order, and therefore each non-empty finite set has   exactly one $\leq_{\mathit{lex}}$-minimal element.
                We conclude that a finite set choice function is definable by $FiniteChoice(X,x) := $``$x$ is the $\leq_{\mathit{lex}}$-minimal element in $X$''.
	\end{claimproof}
	
	Let $FiniteChoice(X,x)$ be an MSO-formula that defines a function as in Claim \ref{claim-finite-choice}.
	We will use formulas $FiniteAntichainSubset(X,Y)$ and $FiniteChoice(X,x)$ to define a choice function by an MSO-formula $Choice(X,x)$ which is the conjunction of the following conditions:
	\begin{enumerate}
		\item $\exists Z:$ ``$Z$ is the set of $\leq$-minimal elements in $X$''
		\item $\exists Y: FiniteAntichainSubset(Z, Y)$
		\item $FiniteChoice(Y,x)$
	\end{enumerate}
	
	For each non-empty set $X$ there is a unique subset $Z \subseteq X$ of the $\leq$-minimal elements in $X$. This set is a non-empty antichain, and therefore $FiniteAntichainSubset(Z,Y)$ returns a finite subset $Y \subseteq Z$. Therefore, $FiniteChoice(Y,x)$ returns an element in $Y$.
	We conclude that $Choice(X,x)$ returns an element $x \in X$ and therefore defines a choice function in the unlabeled full-binary tree, in contradiction to Theorem \ref{theorem-MSO-choice-function}.
\end{proof}


\section{Simple Properties of Automata and Languages} \label{sect:simple}
In this section some simple lemmas are collected.

\begin{lem} \label{lemma-ambiguity-of-union-and-intersection}
	Let $\cA_1 = (Q_1,\Sigma_1, Q^1_{I_1}, \delta_1, \mathbb{C}_1)$ and $\cA_2 = (Q_2,\Sigma_2, Q^2_{I_1}, \delta_2, \mathbb{C}_2)$ be two PTA. Then:
	\begin{enumerate}
		\item There exists an automaton $\cB$ such that $L(\cB) = L(\cA_1) \cup L(\cA_2)$ and for each $t \in L(\cA_1) \cup L(\cA_2)$, $|ACC(\cB, t)| \leq |ACC(\cA_1, t)| + |ACC(\cA_2, t)|$.
		
		\item There exists an automaton $\cB$ such that $L(\cB) = L(\cA_1) \cap L(\cA_2)$ and for each $t \in L(\cA_1) \cap L(\cA_2)$, $|ACC(\cB, t)| \leq |ACC(\cA_1, t)| \cdot |ACC(\cA_2, t)|$.
	\end{enumerate}
\end{lem}
\begin{proof}
	(1) Assume that $Q_1$ and $Q_2$ are disjoint, and let $\cB := (Q_1 \cup Q_2,\Sigma_1 \cup \Sigma_2, Q^1_I \cup Q^2_I, \delta_1 \cup \delta_2, \mathbb{C}_1 \cup \mathbb{C}_2)$. It is clear that $L(\cB) = L(\cA_1) \cup L(\cA_2)$.
	
	Let $t \in L(\cB)$. By definition of $\cB$, for each $\phi \in ACC(\cB, t)$ we either have $\phi \in ACC(\cA_1, t)$ or $\phi \in ACC(\cA_2, t)$. Therefore, we obtain $|ACC(\cB, t)| = |ACC(\cA_1, t)| + |ACC(\cA_2, t)|$.
	
	(2)
	It is easy to verify that there is an MSO-formula over $\omega$-words that holds for $w= (c_1, c'_1), \dots, (c_i, c'_i), \dots \in (Image(\mathbb{C}_1) \times Image(\mathbb{C}_2))^\omega$ iff the maximal color that appears infinitely often in the first coordinate of $w$ and the maximal color that appears infinitely often in the second coordinate of $w$ are both even. Therefore (by McNaughton's Theorem \cite{mcnaughton1966testing}) there is a deterministic PWA $\cD = (Q_\cD, \Sigma_\cD, q_I^\cD, \delta_\cD, \mathbb{C}_\cD)$ over alphabet $\Sigma_\cD = Image(\mathbb{C}_1) \times Image(\mathbb{C}_2)$ such that $w \in L(\cD)$ iff the maximal color that appears infinitely often in the first coordinate of $w$ and the maximal color that appears infinitely often in the second coordinate of $w$ are both even.
	
	We will use the automata $\cA_1, \cA_2$ and $\cD$ to define a PTA $\cB := (Q_\cB, \Sigma_\cB, Q_I^\cB, \delta_\cB, \mathbb{C}_\cB)$ which accepts $L(\cA_1) \cap L( \cA_2)$.
	
	\begin{itemize}
		\item $Q_\cB = Q_1 \times Q_2 \times Q_\cD$
		\item $\Sigma_\cB:=\Sigma_1 \cap \Sigma_2$
		\item $Q_I^\cB := Q_I^1 \times Q_I^2 \times \{q_I^\cD\}$
		\item $((q,p,s), a, (q_1, p_1, s_1), (q_2, p_2, s_2)) \in \delta_\cB$ iff $(q,a,q_1,q_2) \in \delta_1$, $(p,a,p_1,p_2) \in \delta_2$, and $s_1 = s_2 = \delta_\cD(s, (\mathbb{C}_1(q), \mathbb{C}_2(p)))$.
		\item $\mathbb{C}_\cB(q_1, q_2, p) := \mathbb{C}_\cD(p)$
	\end{itemize}
It is easy to verify that 	$L(\cB) = L(\cA_1) \cap L(\cA_2)$.

	Assume, for the sake of contradiction, that there exists $t$ such that $|ACC(\cB, t)| > |ACC(\cA_1, t)| \cdot |ACC(\cA_2, t)|$. Since $\cD$ is deterministic, it follows that there is a computation in $ACC(\cB, t)$ such that either the projection of the first coordinate of $\phi$ on $Q_1$, denoted $\phi_1$, is not in $ACC(\cA_1, t)$ or the projection of the second coordinate of $\phi$ on $Q_2$, denoted $\phi_2$, is not in $ACC(\cA_2, t)$. Assume w.l.o.g. that $\phi_1 \notin ACC(\cA_1, t)$. Therefore, there is a tree branch $\pi = v_0, v_1, \dots$ such that the maximal color that $\mathbb{C}_1$ assigns to the states that occurs infinitely often in $\phi_1(\pi)$ is odd. By definition of $\cD$ we conclude that $w := (c_0, c_0'), (c_1, c_1'), \dots \notin L(\cD)$, where $c_i := \mathbb{C}_1(\phi_1(v_i))$ and $c_i' := \mathbb{C}_2(\phi_2(v_i))$. Hence, by definition of $\cB$ we conclude that the sequence of colors that $\mathbb{C}_\cB$ assigns to the states $\phi(\pi)$ is exactly $w$, and therefore $\phi \notin ACC(\cB, t)$~-- a contradiction.
\end{proof}

From Lemma~\ref{lemma-ambiguity-of-union-and-intersection}, we obtain:
\begin{cor} \label{corollary-ambiguity-closure}
	Boundedly, finitely and countably ambiguous tree languages are closed under finite union and intersection.
\end{cor}
We often use implicitly the following simple Lemma.
\begin{lem}[Grafting]\label{lem:triv-composition}
	Let $\cA$ be an automaton, $t$, $t_1$ trees, $v \in \{l,r\}^*$ and $\phi\in ACC(\cA,t)$,
	and $\phi_1\in ACC(\cA_q,t_1)$.
	If $\phi(v)=q$, then $\phi [ \phi_1/v]$ is an accepting computation of $\cA$ on $t[t_1/v]$.
\end{lem}
A similar lemma holds for general grafting. As an immediate consequence, we obtain the following lemma:
\begin{lem} \label{lemma-A_q-ambiguity}
$da(\cA) \geq da(\cA_q)$ for every useful state $q $ of $\cA$.
\end{lem}
%

\begin{cor} \label{corollary-A_Q-ambiguity}
	Let $\cA$ be a boundedly (respectively, finitely, countably) ambiguous PTA with a set $Q$ of useful states, and let $Q'\subseteq Q$. Then $\cA_{Q'}$ is boundedly (respectively, finitely, countably) ambiguous.
\end{cor}

\begin{lem} \label{lemma-subset-language-of-different-ambiguity}
	Let $L_1$ and $L_2$ be two tree languages such that $da(L_1) \neq da(L_2)$ and $L_1 \subseteq L_2$. Then, there exists a tree $t \in L_2 \setminus L_1$.
\end{lem}
\begin{proof} 
	The lemma follows immediately, since otherwise we have $L_1 = L_2$ and therefore $da(L_1) = da(L_2)$, in contradiction to $da(L_1) \neq da(L_2)$.
\end{proof}

\begin{lem} \label{lemma-equivalent-automaton-with-one-initial-state}
	Let $\cA = (Q,\Sigma, Q_I, \delta, \mathbb{C})$ be a PTA. Then, there exists a PTA $\cB = (Q_\cB, \Sigma, \{q_I^\cB\},\allowbreak \delta_\cB, \mathbb{C})$ with single initial state such that $L(\cB) = L(\cA)$, and $da(\cB) \leq da(\cA)$.
\end{lem}
\begin{proof} 
	Let $Q_\cB := Q \cup \{q_I^\cB\}$ and $\delta_\cB := \delta_\cA \cup \{(q_I^\cB,a,q_l,q_r) \mid q_I \in Q_I$ and $(q_I,a,q_l,q_r) \in \delta\}$. It is easy to see that $L(\cB) = L(\cA)$.
	
	Let $t \in L(\cA)$, and let $g_t$ be a function from $ACC(\cA, t)$ to $ACC(\cB, t)$ that maps each computation $\phi \in ACC(\cA, t)$ to a computation $\phi'$ that assigns $q_I^\cB$ to node $\epsilon$, and $\phi(v)$ to other nodes. It is easy to see that $\phi' \in ACC(\cB, t)$, and that $g_t$ is surjective, and therefore $\forall t: |ACC(\cA, t)| \geq |ACC(\cB, t)|$, as requested.
\end{proof}

\begin{defi}[Moore machine]
	A Moore machine is a tuple $M = (\Sigma, \Gamma, Q, q_I, \delta, out)$, where $\Sigma$ is a finite input alphabet, $Q$ is a finite set of states, $ q_I\in Q$ is an initial state, $\delta: Q\times \Sigma \rightarrow Q$ is a transition function,
	$\Gamma$ is an output alphabet, and $out: Q \rightarrow \Gamma$ is an output function.
	
	Define $\widehat{\delta} : \Sigma^* \rightarrow Q$ by $\widehat{\delta}(\epsilon) := q_I$ and $\widehat{\delta}(w) := \delta(\widehat{\delta}(w'), a)$ for $w = w' \cdot a$ where $w' \in \Sigma^*$ and $a \in \Sigma$.
	We say that a function $F: \Sigma^* \rightarrow \Gamma$ is definable by a Moore machine if there is a Moore machine $M$ such that $F(w) = out(\widehat{\delta}(w))$ for all $w \in \Sigma^*$.
\end{defi}
\begin{defi}
	Let $F: \Sigma_1^* \rightarrow \Sigma_2$ be a function definable by a Moore machine, and let $t_1 \in T^\omega_{\Sigma_1}$. We define $t_2 := \widehat{F}(t_1)$ as a tree in $T^\omega_{\Sigma_2}$ such that $t_2(v) := F(t_1(v_1) \cdot \dots \cdot t_1(v_k))$ where $v_1, v_2, \dots, v_k$ is the path from the root to $v$.
	
	For a tree language $L \subseteq T^\omega_{\Sigma_1}$, we define $\widehat{F}(L) := \{\widehat{F}(t) \mid t \in L \} \subseteq T^\omega_{\Sigma_2}$.
\end{defi}
\begin{lem}[Reduction] \label{lemma-F-reduction}
	Let $L_1$ and $L_2$ be regular tree languages over alphabets $\Sigma_1$ and $\Sigma_2$, respectively. Let $F: \Sigma_1^* \rightarrow \Sigma_2$ be a function definable by a Moore machine.
	Assume that for each $t \in T^\omega_{\Sigma_1}$, $t \in L_1$ iff $\widehat{F}(t) \in L_2$. Then $da(L_1) \leq da(L_2)$.
\end{lem}
\begin{proof}
	Let $\cA_2 = (Q_2 ,\Sigma_2, Q^2_I, \delta_2, \mathbb{C}_2)$ such that $\cA_2$ accepts $L_2$ and $da(\cA_2) = da(L_2)$.
	
	Let $M = (\Sigma_1, \Sigma_2, Q_M, q_I^M, \delta_M, out_M)$ be a Moore machine defining $F$. We will use $\cA_2$ and $M$ to define an automaton $\cA_1 = (Q_1 ,\Sigma_1, Q^1_I, \delta_1, \mathbb{C}_1)$ such that $t \in L(\cA_1)$ iff $\widehat{F}(t) \in L(\cA_2)$, by:	
	\begin{itemize}
		\item $Q_1 := Q_2 \times Q_M$
		\item $Q^1_I := Q^2_I \times \{q^M_I\}$
		\item $((q, p), a, (q_1, p_1), (q_2, p_2)) \in \delta_1$ iff $p_1 = p_2 = \delta_M(p, a)$ and $(q, out_M(p), q_1, q_2) \in \delta_2$
		\item $\mathbb{C}_1(q,p) := \mathbb{C}_2(q)$
	\end{itemize}
	
	First notice that $\forall t \in T^\omega_\Sigma: t \in L(\cA_1) \Leftrightarrow \widehat{F}(t) \in L(\cA_2) \Leftrightarrow \widehat{F}(t) \in L_2 \Leftrightarrow  t \in L_1$, and therefore $L(\cA_1) = L_1$ as needed.
	
	Let $\phi \in ACC(\cA_1, t)$, and define a computation $\phi'$ by $\phi'(v) = q_1$ for $\phi(v) = (q_1, q_2) \in Q_2 \times Q_M$. It is easy to see that $\phi' \in ACC(\cA_2, \widehat{F}(t))$ and since $M$ is deterministic, we conclude that $|ACC(\cA_1, t)| \leq |ACC(\cA_2, \widehat{F}(t)|)|$, and therefore $da(\cA_1) \leq da(\cA_2)$.
	
	We conclude that $da(L_1) \leq da(\cA_1) \leq da(\cA_2) = da(L_2)$, as requested.
\end{proof}
Let us state  another well-known characterization of regular trees.
\begin{fact}\label{fact:sect3}
 A tree $t$ is regular iff its labeling $t~:~\{l,r\}^*\rightarrow \Sigma$ is definable by a Moore machine.
\end{fact}

\section{Not-Finitely Ambiguous Languages} \label{sect:not-finite}
We provide here sufficient conditions for a language to be not finitely ambiguous. These conditions will allow us to present some natural languages which are countably ambiguous and not finitely ambiguous, proving Theorem \ref{th:hierarchy}(3). In addition, these results are used
in Sects. \ref{sect:k-ambig}-\ref{sect:uncount} where it is proved that for every $k>1$ there is a   language of ambiguity degree equal to $k$
and there are languages with finite and uncountable degrees of ambiguity.

First, we state our main technical result~-- Proposition \ref{prop:main-new}. Then, we derive some consequences. Finally, a proof of Proposition \ref{prop:main-new} is given. Our proof relies on the fact that there is no MSO-definable function that assigns to every non-empty antichain $Y$ a finite non-empty subset $X\subseteq Y$
(Lemma~\ref{lemma-definable-selection-functions}), and our proof techniques refine the proof techniques of \cite{carayol2010choice}.

Recall that for trees $t$ and $t'$ and an antichain $Y$, we denote by $t[t'/Y]$ the tree obtained from $t$ by grafting $t'$ at every node in $Y$.
\begin{prop} \label{prop:main-new} Let $t_0$ and $t_1$ be regular trees and $L$ be a regular language such that $t_0\not\in L$ and
$t_0[t_1/Y]\in L$ for every
non-empty antichain $Y$. Then $L$ is not finitely ambiguous.
\end{prop}
\saveCounter
\begin{defi}
	For a tree language $L$ over alphabet $\Sigma$, we denote by $Subtree(L)$ the tree language $\{t \in T^\omega_\Sigma \mid \exists t' \in L \: \exists v: t'_{\geq v} = t\}$.
\end{defi}

\begin{cor} \label{cor-countable-complement-not-finitely-ambiguous}
	Let $L$ be a non-empty regular language over an alphabet $\Sigma$ such that $Subtree(L) \neq T^\omega_\Sigma$.
	Then, the complement of $L$ is not finitely ambiguous.
\end{cor}
\begin{proof}
 Let $L$ be as in Corollary~\ref{cor-countable-complement-not-finitely-ambiguous}.
We claim that there are regular $\Sigma$-labeled trees $t_0\in L$ and $t_1\not\in Subtree(L)$. Indeed, by Rabin's basis theorem there is a regular $t_0\in L$. Since $L$ is regular, there is an automaton $\cB = (Q, \Sigma, \{q_I\}, \delta, \mathbb{C})$ (with only useful states) that accepts $L$. It is clear that $\cB_Q$ accepts $Subtree(L)$, and therefore $Subtree(L)$ is regular.
The complement of $Subtree(L)$ is regular (as the complement of a regular language) and non-empty (since $Subtree(L) \neq T^\omega_\Sigma$), and therefore contains a regular tree $t_1$ (by Rabin's basis theorem). Note that $t_0[ t_1/Y]\not\in {L}$ for every non-empty antichain $Y$.

The complement of $L$ satisfies the assumption of Proposition \ref{prop:main-new}. Therefore, it is not finitely ambiguous.
\end{proof}

\begin{cor}[Not finitely ambiguous languages] \label{corollary-not-finitely-ambiguous-languages}
	The following languages are not finitely ambiguous:
	\begin{enumerate}
		\item The complement of a non-empty regular countable tree language.
		\item \label{corollary-singleton-complement} The complement of a regular language that contains a single tree.
		\item \label{corollary-exists-ai-not-finitely-ambiguous}
		The language $L_{\exists a_1} := \{t \in T_\Sigma^\omega \mid t$ has
		at least one node labeled by $a_1 \}$ over alphabet $\Sigma = \{a_1, \dots, a_m, c\}$.
	\end{enumerate}
\end{cor}
\begin{proof}
	%
	(1) Let $L$ be a non-empty regular countable tree language. Every tree has countably many subtrees, and since $L$ is countable we conclude that $Subtree(L)$ is countable. Therefore, $Subtree(L)$ does not contain all trees. By Corollary~\ref{cor-countable-complement-not-finitely-ambiguous}, we conclude that $\overline{L}$ is not finitely ambiguous.
	
	(2)	Follows immediately from (1).
	
	(3) By the definition of $L_{\exists a_1}$ we have $L_{\exists a_1} \cap T^\omega_{\{c,a_1\}} = T^\omega_{\{c,a_1\}} \setminus \{t_c\}$, and therefore by (2), $L_{\exists a_1} \cap T^\omega_{\{c,a_1\}}$ is not finitely ambiguous.
	It is easy to see that $T^\omega_{\{c,a_1\}}$ is unambiguous (since there is a deterministic automaton that accepts it). Therefore, by Corollary~\ref{corollary-ambiguity-closure} we conclude that $L_{\exists a_1}$ is not finitely ambiguous.
\end{proof}
It is easy to prove that the complement of every finite language (i.e. a language which consists of finitely many trees) is countably ambiguous. Therefore, we obtain:
\begin{cor} \label{corollary-not-finitely-countable-ambiguous-languages}	
	If $L$
	is regular and its
	complement is finite and non-empty, then $da(L)=\aleph_0$. 
\end{cor}
\begin{proof}[Proof of Corollary~\ref{corollary-not-finitely-countable-ambiguous-languages}]
	We first prove the following claim:
	\begin{claim} \label{claim-singletone-complement-is-regular}
		Let $L$ be a regular tree language containing a single tree. Then $\overline{L}$ is countably ambiguous.
	\end{claim}
	\begin{claimproof}
		Assume that $L = \{t\}$.
		$L$ is a regular language, and therefore $t$ is regular. By Fact \ref{fact:sect3}  there is a Moore machine $M = (\{l,r\}, \Sigma, Q_M, q_I^M, \delta_M, out_M)$ such that for each $v \in \{l,r\}^*$, $out(\widehat{\delta}(v)) = \sigma$ iff $t(v) = \sigma$ (that is, $M$ defines the function $t: \{l,r\}^* \rightarrow \Sigma$).
		
		We will use $M$ to construct a countably ambiguous automaton $\cA$ that accepts $\overline{L}$ by guessing a node $v \in \{l,r\}^*$ such that $t(v) \neq t'(v)$ for each tree $t' \in \overline{L}$.
		
		Let $\cA := (Q_\cA,\Sigma, Q_I, \delta, \mathbb{C})$ such that:
		\begin{itemize}
			\item $Q_\cA := \{q, q'\} \times Q_M$
			\item $Q_I := \{(q', q_I^M)\}$
			\item $\delta$ is defined by:
			\begin{itemize}
				\item $((q,p), a, (q,p'), (q,p'')) \in \delta$ iff $\delta_M(p, l) = p'$, $\delta_M(p, r) = p''$
				\item $((q',p), a, (q,p'), (q,p'')) \in \delta$ iff $\delta_M(p, l) = p'$, $\delta_M(p, r) = p''$ and $out(p) \neq a$
				\item $((q',p), a, (q',p'), (q,p'')), ((q',p), a, (q,p'), (q',p'')) \in \delta$ iff $\delta_M(p, l) = p'$, $\delta_M(p, r) = p''$ and $out(p) = a$.
			\end{itemize}
			\item $\forall p \in Q_M: \mathbb{C}(q,p) := 0$ and $\mathbb{C}(q',p) := 1$
		\end{itemize}
		
		By definition of $\cA$, it is clear that $t' \in L(\cA)$ iff there is a node $v$ such that $t'(v) \neq t(v)$, and therefore $t' \in L(\cA)$ iff $t' \neq t$.
		
		For each computation $\phi$ of $\cA$ on $t'$, the $Q_M$ component is determined deterministically by $M$ and $t$. If $\phi$ is accepting, there are finitely many nodes $v$ such that the first component of $\phi(v)$ is $q'$~-- otherwise, there would be a branch where the maximal color assigned infinitely often by $\mathbb{C}$ is odd, in contradiction to $\phi$ being an accepting computation. Therefore, there are countably many accepting computations on each tree $t' \in L(\cA)$, and $\cA$ is countably ambiguous.
	\end{claimproof}
	
	$L$ is finite and therefore there are $t_1, \dots, t_k \in T^\omega_\Sigma$ such that $L = \{t_1, \dots, t_k\}$.
	A finite tree language does not contain a non-regular tree, and therefore $t_1, \dots, t_k$ are regular.
	By Claim \ref{claim-singletone-complement-is-regular}, for each tree $t_i \in L$, there is a countably ambiguous automaton $\cA_i$ such that $t \in L(\cA_i)$ iff $t \neq t_i$. Notice that $\overline{L} = L(\cA_1) \cap \dots, \cap L(\cA_k)$, and therefore by Lemma~\ref{lemma-ambiguity-of-union-and-intersection} we conclude that $\overline{L}$ is countably ambiguous.	
\end{proof}

\loadCounter
\paragraph{On the Proof of Proposition \ref{prop:main-new}}
In the rest of this section, Proposition \ref{prop:main-new} is proved.
Let us sketch some ideas of the proof.
For a language $L$, as in Proposition \ref{prop:main-new}, and any non-empty antichain $Y$   we show that
 if $\cA$ does not accept $t_0$ and accepts $t:=t_0[t_1/Y]$, then every $\phi\in ACC(\cA,t)$ chooses (in an MSO-definable way) an element from $Y$. Hence, the computations in $ACC(\cA,t)$ choose together a subset $X$ of $Y$ of cardinality $\leq |ACC(\cA,t)|$ (each computation chooses a single element). Therefore, if $\cA$ accepts $ {L}$ and
is finitely ambiguous, then $X$ is finite~-- a contradiction to Lemma~\ref{lemma-definable-selection-functions}.
To implement this plan,
in Subsection~\ref{subsect:membership} we recall a game theoretical
interpretation of ``a tree is accepted by an automaton.'' Then, in Subsection~\ref{subsect:mso-defin-game} we analyze which concepts related to these games are MSO-definable. Finally, in Subsection~\ref{subsect:finish}, the proof is completed.
\subsection{Membership Game}\label{subsect:membership}
Let $\cA = (Q, \Sigma, \{q_I\}, \delta, \mathbb{C})$ be a PTA, and let $t$ be a $\Sigma$-labeled tree.	
A two-player game $G_{t, \cA}$ (called a ``membership game'') between Automaton and Pathfinder is defined as follows. The positions of Automaton are $\{l,r\}^* \times Q$, and the positions of Pathfinder are $\{l,r\}^* \times Q \times Q$. The initial position is $(\epsilon, q_I)$.

From a position $(v, q) \in \{l,r\}^* \times Q$ Automaton chooses a tuple $(q_l,q_r) \in Q \times Q$ such that $\exists a \in \Sigma: (q, a, q_l, q_r) \in \delta$, and moves to the position $(v, q_l, q_r)$. From a position $(v, q_l, q_r) \in \{l,r\}^* \times Q \times Q$ Pathfinder chooses a direction $d \in \{l,r\}$, and moves to the position $(v \cdot d, q_d)$.

We define a \textbf{play} $\overline{s} := e_0, d_0, e_1, d_1, \dots, e_i, d_i, \dots \in (Q \times Q \times \{l,r\})^\omega$ as an infinite sequence of moves, corresponding to the choices of Automaton and Pathfinder from the initial position.
We say that the move $e_i = (q_l, q_r)$ from position $(v, q)$ is \textbf{invalid} for Automaton if $(q, t(v), q_l, q_r) \notin \delta$.

A \textbf{strategy} for a player in $G_{t, \cA}$ is a function that determines the next move of the player based on previous moves of both players.

A \textbf{positional strategy} for a player in $G_{t, \cA}$ is a strategy that determines the next move of the player based only on the current position. A positional strategy for Automaton is a function $str: \{l,r\}^* \times Q \rightarrow Q \times Q$, and a positional strategy for Pathfinder is a function $STR: \{l,r\}^* \times Q \times Q \rightarrow \{l,r\}$.

Let $\mathbb{C}_G$ be a coloring function that maps each position in $G_{t, \cA}$ to a color in $\nat$. We define $\mathbb{C}_G(v, q) := \mathbb{C}(q)$ for Automaton's positions, and $\mathbb{C}_G(v,q_l ,q_r) := 0$ for Pathfinder's positions.

For each play $\overline{s}$ define $\pi_{\overline{s}}$ as the infinite sequence of positions corresponding to the moves in $\overline{s}$.	A play $\overline{s}$ is winning for Automaton iff $\overline{s}$ does not contain an invalid move for Automaton, and the maximal color that $\mathbb{C}_G$ assigns infinitely often to the positions in $\pi_{\overline{s}}$ is even. Since all Pathfinder's positions are colored by $0$, it is sufficient to consider the coloring of Automaton's positions in $\pi_{\overline{s}}$.

We say that a play is consistent with a strategy of a player if all moves of the player are according to the strategy. A \textbf{winning strategy} for a player is a strategy such that each play that is consistent with the strategy is winning for the player.

Parity games are positionally determined \cite{emerson1991tree}, i.e., for each parity game, one of the players has a positional winning strategy. Therefore, if a player has a winning strategy, then he has a positional winning strategy. Additionally, if a positional strategy of a player wins against all positional strategies of the other player, then it is a winning strategy.

We recall standard definitions and facts about the connections between games and tree automata \cite{gurevich1982trees,PinPerrin}.

Let $\phi: \{l,r\}^* \rightarrow Q$ be a function such that $\phi(\epsilon) = q_I$ and $\forall v \in \{l,r\} : \exists a \in \Sigma: (\phi(v), a, \phi(v \cdot l), \phi(v \cdot r)) \in \delta$. We define a positional strategy $str_\phi: \{l,r\}^* \times Q \rightarrow Q \times Q$ for Automaton, by $str_\phi(v,q) := (\phi(v \cdot l), \phi(v \cdot r))$.
Conversely, for each positional strategy $str: \{l,r\}^* \times Q \rightarrow Q \times Q$ of Automaton we construct a function $\phi_{\mathit{str}}: \{l,r\}^* \rightarrow Q$ by $\phi(\epsilon) := q_I$ and for all $v \in \{l,r\}^*$ we set $\phi(v \cdot l) := q_l$, and $\phi(v \cdot r) := q_r$ where $str(v, \phi(v)) = (q_l, q_r)$.

\begin{claim} \phantomsection \label{claim-computation-to-strategy}
	\begin{enumerate}
		\item \label{bullet-computation-on-vi} Let $\overline{s}$ be a play that is consistent with $str_\phi$, and let $(v_i, q_i)$ be the $i$-th position of Automaton in $\pi_{\overline{s}}$. Then, $\phi(v_i) = q_i$.
		
		\item If $\phi \in ACC(\cA, t)$, then $str_\phi$ is a positional winning strategy for Automaton.
		
                \item If $str$ is a positional winning strategy for Automaton, then $\phi_{\mathit{str}} \in ACC(\cA, t)$.
	\end{enumerate}
\end{claim}
\begin{claimproof}
	(1) We will prove by induction on $i$. For $i=0$ we have $(v_0, q_0) = (\epsilon, q_I)$ (by definition of $G_{t, \cA}$), and indeed $\phi(v_0) = \phi(\epsilon) = q_I$. Assume the claim holds for $i = k$ and we prove for $i = k+1$.
	
	Let $d \in \{l,r\}$ be the $i$-th move of Pathfinder in $\overline{s}$. By definition of $G_{t, \cA}$ we have $v_{i+1} = v_i \cdot d$, and $q_{i+1} = q_d$, where $str_\phi(v_i, q_i) = (q_l, q_r)$.
	
	By definition $str_\phi$ we have $(q_l, q_r) = (\phi(v_i \cdot l), \phi(v_i \cdot l))$, and therefore $q_{i+1} = \phi(v_i \cdot d) = \phi(v_{i+1})$, as requested.		
	
	(2) and (3) are well known results about membership games \cite{PinPerrin}. 
\end{claimproof}

The next claim describes what happens when Pathfinder plays his winning strategy in~$G_{t, \cA}$ against an Automaton's winning strategy
in $G_{t', \cA}$ (for $t'\neq t$). 	
\begin{claim} \label{claim-automaton-makes-invalid-move}
	Assume $t \notin L(\cA)$ and let $\phi$ be an accepting computation of $\cA$ on a tree $t'$, and $STR$ be a winning strategy of Pathfinder in $G_{t, \cA}$. Let $\overline{s} := e_0, d_0, e_1, d_1, \dots, e_i, d_i, \dots$ be the play that is consistent with $str_\phi$ and $STR$. Then, there is $i \in \nat$ such that $e_i$ is an invalid move for Automaton in $G_{t, \cA}$. Moreover, if $e_i$ is the first invalid move for Automaton in $\overline{s}$, then $t(v) \neq t'(v)$ for
	$v:=d_0 \dots d_{i-1}$.
\end{claim}
\begin{claimproof}
	Assume, for the sake of contradiction, that $\overline{s}$ does not contain an invalid move for Automaton, and let $(v_i, q_i)$ be the $i$-th position of Automaton in $\pi_{\overline{s}}$.
	By definition of $G_{t, \cA}$ it is easy to see that $\pi = v_0, \dots, v_i, \dots$ is a branch in the full-binary tree. Since $\phi$ is an accepting computation of $\cA$ on $t'$, we conclude that the maximal color that $\mathbb{C}$ assigns infinitely often to states in $\phi(\pi)$ is even. By Claim \ref{claim-computation-to-strategy}(\ref{bullet-computation-on-vi}) we have $\phi(v_i) = q_i$, and therefore $\phi(\pi) = q_0 \dots q_i \dots$. By the definition of $\mathbb{C}_G$ we have $\mathbb{C}_G(v_i, q_i) = \mathbb{C}(q_i)$ and we conclude that the maximal color that $\mathbb{C}$ assigns infinitely often in $\pi_{\overline{s}}$ is even, and therefore the play is winning for Automaton~-- a contradiction to $STR$ being a winning strategy of Pathfinder.
	
	Therefore, Automaton makes an invalid move in $\overline{s}$. Let $e_i = (q_l, q_r)$ be the first invalid move of Automaton in $\overline{s}$. Since $e_i$ is invalid we have $(q_i, t(v_i), q_l, q_r) \notin \delta$, and by definition of $str_\phi$ we obtain $(q_l ,q_r) = (\phi(v_i \cdot l), \phi(v_i \cdot r))$. Since $\phi(v_i) = q_i$ we have $(\phi(v_i), t(v_i), \phi(v_i \cdot l), \phi(v_i \cdot r)) \notin \delta$. $\phi$ is a computation of $\cA$ on $t'$ and therefore $(\phi(v_i), {t'}(v_i), \phi(v_i \cdot l), \phi(v_i \cdot r)) \in \delta$, and we conclude that $t(v_i) \neq {t'}(v_i)$. Notice that by the definition of $G_{t, \cA}$ we have $v_i = d_0 \dots d_{i-1}$, and the claim follows.
\end{claimproof}
\subsection{MSO-definability}\label{subsect:mso-defin-game}
Throughout this section we will use the following conventions and terminology.
%
%
\begin{description}[labelwidth=*]
	\item[Positional Pathfinder strategies as labeled trees]
	We will define a positional strategy $STR$ for Pathfinder as a function in $\{l,r\}^* \times Q \times Q \rightarrow \{l,r\}$. Hence, it can be considered
	as a $Q\times Q \rightarrow \{l,r\}$ labeled tree. Below we will not distinguish between a positional Pathfinder's strategy and the corresponding
	$Q\times Q \rightarrow \{l,r\}$ labeled full-binary tree. In particular, we call such a strategy regular, if the corresponding tree
	is regular.
	\item[MSO-definability] We will use ``MSO-definable'' for ``MSO-definable in the unlabeled full-binary tree.''
\end{description}
The rest of the proof deals with MSO-definability.
By Claim \ref{claim-automaton-makes-invalid-move}, there is a function $Invalid_\cA(\phi, STR, t, v)$ that, for every accepting computation $\phi$ of $\cA$ on $t'$, returns a node $v$ such that ${t'}(v) \neq t(v)$. This function depends on the strategy $STR$ of Pathfinder.
The restriction of $Invalid_\cA$ to the
Pathfinder positional winning strategies in $G_{t, \cA}$ is MSO-definable (with parameters $t$ and $STR$) by the following formula $ Leads_\cA(\phi, STR, t, v)$, that describes in MSO the play of $\phi$ against $STR$ up to the first invalid move of Automaton (at the position $(v,\phi(v)$).

Define $Leads_\cA(\phi, STR, t, v) $ as the conjunction of:
\begin{enumerate}
	\item $\phi(\epsilon) = q_I $~--the play starts from the initial position.
	\item $\forall u < v: ((\phi(u), t(u), \phi(u \cdot l), \phi(u \cdot r)) \in \delta $~-- all Automaton's moves at the positions $(u,q)$, where $u$ is an ancestor of $v$ respect $\delta$. (By Claim \ref{claim-computation-to-strategy}(1), in any play consistent with $\phi$, Automaton
	can reach only the positions of the form $(u,\phi(u))$).
	\item $(\phi(v), {t}(v), \phi(v \cdot l), \phi(v \cdot r)) \notin \delta$~-- the Automaton move at $(v,\phi(v))$ is invalid.
	\item $ \forall u <v: (STR(u, \phi(u \cdot l), \phi(u \cdot r)) = l) \leftrightarrow u \cdot l \leq v))$~-- the Pathfinder moves $d_0\dots d_j\dots$ are consistent with $STR$ and are along the path from the root to $v$, i.e., 
	$d_0d_1 \dots d_j \leq v$.
\end{enumerate}
To sum up, we have the following claim:
\begin{claim} \label{claim:summary}
	$Leads_\cA(\phi, STR, t, v) $ defines a function that, for every tree $t\not\in L(\cA)$, every Pathfinder's positional (in $G_{t, \cA}$) winning strategy
	$STR$, and every $\phi \in ACC(\cA,t')$, returns a node $v$ such that $t(v)\neq t'(v)$. %
\end{claim}
\noindent
Claim \ref{claim:summary} plays a crucial role in our proof. It is instructive to compare it with Theorem \ref{theorem-MSO-choice-function}
which implies that there is no MSO-definable function $F(t,D,v)$ that for a tree $t\neq t'$ and $D:=\{u\mid t(u)\neq t'(u)\}$ returns a node $v$ such that $t (v)\neq t'(v)$.

The following claim is folklore. Due to the lack of references, it is proved in the Appendix.

\vspace{0pt}
\begin{restatable}{claim}{claimregularpositionalwinningstrategy} \label{claim-regular-positional-winning-strategy}
	Let $t_0$ be a regular tree such that $t_0 \notin L(\cA)$. Then, Pathfinder has a regular positional winning strategy in $G_{t_0, \cA}$.
\end{restatable}
Let $t_0$ be a regular tree such that $t_0 \notin L(\cA)$. By Claim \ref{claim-regular-positional-winning-strategy} there is a regular
positional winning strategy $\widehat{STR}$ of Pathfinder in $G_{t_
	0, \cA}$. Now, we can substitute $\widehat{STR}$ and $t_0$ for arguments $STR$ and $t$ of $Leads_\cA$ and obtain the following Proposition:
\setcounter{claimCounter}{\value{claim}}
\addtocounter{thm}{4}
\begin{prop} \label{claim-leads2}
	For every regular tree 	$t_0 \notin L(\cA)$ and a regular positional winning strategy $\widehat{STR}$ for Pathfinder in $G_{t_0, \cA}$, there is an MSO-definable function that, for each accepting computation $\phi$ of $\cA$ on $t'$, returns a node $v$ such that $t_0(v) \neq t'(v)$.
\end{prop}
\begin{claimproof} Let $\psi_{t_0}(\sigma)$ and $\psi_{\widehat{STR}}(STR) $ be MSO-formulas that define $t_0$ and
	$\widehat{STR}$.
	Then, by Claim \ref{claim:summary}, $\exists \sigma \exists STR: \psi_{t_0}(\sigma) \wedge \psi_{\widehat{STR}}(STR) \wedge Leads_\cA(\phi, STR, \sigma, v)$ defines
	such a function.
\end{claimproof}
Let us continue with  the proof of 	 Proposition \ref{prop:main-new}.
Recall that for trees $t$ and $t'$ and an antichain $Y$, we denote by $t[t'/Y]$ the tree obtained from $t$ by grafting $t'$ at every node in $Y$.
%
%

%
\addtocounter{thm}{-5}%
\setcounter{claim}{\value{claimCounter}}
\begin{claim} \label{claim-t0<Y>t1-is-regular}
	Let $t_0$ and $t_1$ be regular trees. Then, there is an MSO-formula $\graft_{t_0,t_1}(Y,\sigma) $ defining a function that
	for every antichain $Y$ returns the tree $t_0[ t_1/Y]$.
\end{claim}
\begin{claimproof}[Proof of Claim \ref{claim-t0<Y>t1-is-regular}]
	$t_0$ and $t_1$ are regular, and therefore there are MSO-formulas $\psi_{t_0}(\sigma)$ and $\psi_{t_1}(\sigma)$ that defines $t_0$ and $t_1$. 
	
	Let $\psi^{\geq y}_{t_1}(y, \sigma)$ be a formula that is obtained from $\psi_{t_1}(\sigma)$ by relativizing the first-order quantifiers to $\geq y$, i.e., by replacing subformulas of the form $\exists x (\dots) $ and $\forall x(\dots)$ by
	$\exists x (x\geq y) \wedge (\dots)$ and $\forall x (x\geq y) \rightarrow (\dots)$. Then, $v,t\models \psi^{\geq y}_{t_1}(y, \sigma)$
	iff $t_{\geq v} = t_1$.
	Hence, $\graft_{t_0,t_1}(Y,\sigma) $ can be defined as the conjunction of:
	\begin{enumerate}
		\item $\exists \sigma_{ 0} \psi_{t_0}(\sigma_{0})\wedge \forall v $~-- ``if no $Y$ node is an ancestor of $v$ then $\sigma(v)=\sigma_{0}(v)$,'' and
		\item $\forall y ( y\in Y) \rightarrow \psi^{\geq y}_{t_1}(y,
                  \sigma)$~-- ``at every node in $Y$ a tree $t_1$ is grafted.''
                  \qedhere
	\end{enumerate}
\end{claimproof}

\subsection{Finishing the Proof of Proposition \ref{prop:main-new} }\label{subsect:finish}
{Now, we} have all the ingredients ready for the proof of Proposition \ref{prop:main-new}.

Let $\cA$ be such that $L(\cA) =  {L}$, and let $\lead $ be a formula that defines the function from Proposition \ref{claim-leads2} ($t_0[ t_1/Y]$ now takes the role of $t'$).

Define a formula: $Choice_{\cA, t_0, t_1, \widehat{STR}}(Y, \phi, y) := y \in Y \wedge \exists v (\lead \wedge v \geq y)$.

\begin{claim} \label{claim-choice-function-with-computation}
	$Choice_{\cA, t_0, t_1, \widehat{STR}}(Y, \phi, y)$ defines a function that for every non-empty antichain $Y$ and an accepting computation $\phi$ of $\cA$ on $t_0 [t_1/Y ]$, returns a node $y \in Y$.
\end{claim}
\begin{claimproof}
	By Proposition \ref{claim-leads2}, $\lead $ returns a node $v$ such that $t_0(v) \neq (t_0[t_1/Y])(v)$. By definition of $t_0 [t_1/Y]$, there is a unique node $y \in Y$ such that $v \geq y$.
\end{claimproof}
Define
$ChooseSubset_{\cA, t_0, t_1, \widehat{STR}}(Y, X) := \forall x: x \in X $ iff the following conditions hold:
\begin{enumerate}
	\item $ x \in Y$ and
	\item $ \exists \sigma$ such that
	\begin{enumerate}
		\item $ \graft_{t_0,t_1}(Y,\sigma)$~-- ``$\sigma = t_0[t_1/Y]$'' and
		\item $\exists \phi AcceptingRun_\cA(\sigma, \phi) \wedge Choice_{\cA, t_0, t_1, \widehat{STR}}(Y, \phi, x)$, {where}
		$ AcceptingRun_\cA(\sigma, \phi)$ defines ``$\phi $ is an accepting computation of $\cA$ on the tree $\sigma$.''
	\end{enumerate}
\end{enumerate}

\begin{claim} \label{claim-choice-function-last}
	$ChooseSubset_{\cA, t_0, t_1, \widehat{STR}}(Y, X) $ defines a function that maps every non-empty antichain $Y$ to a non-empty subset $X\subseteq Y$. Moreover, $|X|\leq |ACC(A, t_0[t_1/Y])|$.
\end{claim}
\begin{claimproof} If $Y$ is non-empty, then $t_0[t_1/Y]\in  {L}$. Hence, $\cA$ has at least one
	accepting computation on $t_0[t_1/Y]$. Therefore, $X$ is non-empty, by Claim \ref{claim-choice-function-with-computation}. The ``Moreover'' part immediately follows from Claim \ref{claim-choice-function-with-computation}.
\end{claimproof}
Let $\cA$ be such that $L(\cA) =  {L}$ and assume, for the sake of contradiction, that $\cA$ is finitely ambiguous. In particular, there are finitely many accepting computations of $\cA$ on $t_0 [t_1/Y]$, and therefore by Claim \ref{claim-choice-function-last}, we conclude that $ChooseSubset_{\cA, t_0, t_1, \widehat{STR}}(Y, X)$ assigns to every non-empty antichain $Y$ a finite non-empty $X\subseteq Y$~-- a contradiction to Lemma~\ref{lemma-definable-selection-functions}.

\section{\texorpdfstring{$k$}{k}-Ambiguous Languages} \label{sect:k-ambig}
In this section we prove that for every $0<k\in \nat$, there is a tree language with the degree of ambiguity equal to $k$.
First, we introduce some notations.
For a letter $\sigma$, we denote by $t_\sigma$, the full-binary tree with all nodes labeled by $\sigma$.
Let $L_{\neg a_1 \vee \dots \vee \neg a_k} := L_{\neg a_1} \cup \dots \cup L_{\neg a_k}$ be a tree language over alphabet $\Sigma_n = \{c, a_1, a_2, ..., a_n\}$, where $L_{\neg a_i} := \{t \in T_{\Sigma_n}^\omega \mid$ no node in $t$ is labeled by $a_i \}$.

%

\begin{prop}
	The degree of ambiguity of $L_{\neg a_1 \vee \dots \vee \neg a_k}$ for $k \leq n$ is $k$. 
\end{prop}
It is easy to see that $L_{\neg a_i}$ are accepted by deterministic PTA. Therefore,
by Lemma~\ref{lemma-ambiguity-of-union-and-intersection}, we obtain that $L_{\neg a_1 \vee \dots \vee \neg a_k}$ is $k$-ambiguous. In the rest of this section we will show that $L_{\neg a_1 \vee \dots \vee \neg a_k}$ is not $(k-1)$-ambiguous. It was shown in \cite{DBLP:conf/csl/BilkowskiS13} that
$L_{\neg a_1 \vee \neg a_2}$ is ambiguous.

\begin{lem} \label{lemma-not-finite}
	Let	$L_{\exists a_1 \wedge \dots \wedge \exists a_m}:=\{t \in T_{\Sigma_n}^\omega \mid$ for every $i\leq m $ there is a node in $t$ labeled by $a_i \}$, and let $L$ be a tree language such that $t_c \notin L$ and $L_{\exists a_1 \wedge \dots \wedge \exists a_m} \cap T^\omega_{\{c, a_1, \dots, a_m\}} \subseteq L$. Then, $L$ is not finitely ambiguous.
\end{lem}
\begin{proof}
	Define a function $F: \Sigma^* \rightarrow \Sigma$ such that
	$F(\sigma_1 \dots \sigma_k) := a_{k-i+1}$
	if there is $i$ such that $\sigma_i = a_1$, for all $j<i: \sigma_j \neq a_1$ and $k-i+1 \leq m$. Otherwise, $F(\sigma_1 \dots \sigma_k) := c$.
	
	It is easy to see that $F$ is definable by a Moore machine, and $\forall t \in T^\omega_\Sigma: t \in L_{\exists a_1}$ iff $\widehat{F}(t) \in L$. Therefore, by Lemma~\ref{lemma-F-reduction} we conclude that $da(L) \geq da(L_{\exists a_1})$. Since $L_{\exists a_1}$ is not finitely ambiguous (by Corollary~\ref{corollary-not-finitely-ambiguous-languages}~(\ref{corollary-exists-ai-not-finitely-ambiguous})), we conclude that $L$ is not finitely ambiguous.
\end{proof}
\paragraph{Notations}
Let $a \in \Sigma$, $t_1 \in T^\omega_\Sigma$ and $t_2 \in T^\omega_\Sigma$. We define $Tree(a,t_1,t_2) \in T^\omega_\Sigma$ as a tree $t$ where $t(\epsilon) = a$, $t_{\geq l} = t_1$ and $t_{\geq r} = t_2$.

%

%
\begin{lem} \label{lemma-at-least-k-ambiguous}
	Let $\cA$ be a finitely ambiguous automaton over alphabet $\Sigma_n$ such that $L(\cA) = L_{\neg a_1 \vee \dots \vee \neg a_k}$ for $k \leq n$. Then $|ACC(\cA, t_c)| \geq k$.
\end{lem}
\begin{proof}
	We will prove by induction on $k$.
	For $k = 1$ the claim holds trivially, since $t_c \in L(\cA)$ implies that $|ACC(\cA, t_c)| \geq 1$.
	
	Assume the claim holds for all $k < m \leq n$ and prove for $k = m$.
	
	
	Let $\cA = (Q,\Sigma, Q_I, \delta, \mathbb{C})$ be a finitely ambiguous automaton that accepts $L_{\neg a_1 \vee \dots \vee \neg a_m}$. Define $R := \{(q_1, q_2) \in Q \times Q \mid \exists q_i \in Q_I: (q_i, c, q_1, q_2) \in \delta)\}$, and let $R[1]$ and $R[2]$ be the projections of the first and second coordinate of $R$ on $Q$, respectively.
	
	Define $Q_{\exists a_m} := \{q \in R[1] \mid L(A_q) \cap L_{\exists a_m} \neq \emptyset \}$, and let
	$Q_{\exists a_m \wedge t_c} := \{q \in Q_{\exists a_m} \mid t_c \in L(\cA_q)\}$ and
	$Q_{\exists a_m \wedge \neg t_c} := Q_{\exists a_m} \setminus Q_{\exists a_m \wedge t_c}$.
	
	By definition of $Q_{\exists a_m \wedge \neg t_c}$ we have $t_c \notin L(\cA_{Q_{\exists a_m \wedge \neg t_c}})$ and therefore $L(\cA_{Q_{\exists a_m \wedge \neg t_c}}) \cap T^\omega_{\{c,a_m\}} \subseteq T^\omega_{\{c,a_m\}} \setminus \{t_c\}$. The language $T^\omega_{\{c,a_m\}} \setminus \{t_c\}$ is not finitely ambiguous by Corollary~\ref{corollary-not-finitely-ambiguous-languages}~(\ref{corollary-singleton-complement}).
	$L(\cA_{Q_{\exists a_m \wedge \neg t_c}})$ is finitely ambiguous (by Corollary~\ref{corollary-A_Q-ambiguity}) and since $T^\omega_{\{c,a_m\}}$ is unambiguous we conclude that $L(\cA_{Q_{\exists a_m \wedge \neg t_c}}) \cap T^\omega_{\{c,a_m\}}$ is finitely ambiguous, by Corollary~\ref{corollary-ambiguity-closure}. Therefore, by Lemma~\ref{lemma-subset-language-of-different-ambiguity}, there is a tree $t' \in T^\omega_{\{c,a_m\}} \setminus \{t_c\} = L_{\exists a_m} \cap T^\omega_{\{c,a_m\}}$ such that $t' \notin L(\cA_{Q_{\exists a_m \wedge \neg t_c}})$, and since $L_{\exists a_m} \cap T^\omega_{\{c,a_m\}} \subseteq L(\cA_{Q_{\exists a_m}}) = L(\cA_{Q_{\exists a_m \wedge t_c}}) \cup L(\cA_{Q_{\exists a_m \wedge \neg t_c}})$ we conclude that $t' \in L(\cA_{Q_{\exists a_m \wedge t_c}})$.
	
	Define $Q' := \{q \in R[1] \mid t' \in L(\cA_q)\}$ and $R' := \{(q_1,q_2) \in R \mid q_1 \in Q'\}$.	
	Since $t' \in L_{\exists a_m} \cap T^\omega_{\{c,a_m\}}$, we conclude that $\{t \in T^\omega_\Sigma \mid Tree(c, t', t) \in L_{\neg a_1 \vee \dots \vee \neg a_m}\} = L_{\neg a_1 \vee \dots \vee \neg a_{m-1}}$. Therefore, $L(\cA_{R'[2]}) = L_{\neg a_1 \vee \dots \vee \neg a_{m-1}}$, and by induction assumption we now obtain that: $|ACC(\cA_{R'[2]}, t_c)| \geq m-1$.
	
	For each computation $\phi \in ACC(\cA_{R'[2]}, t_c)$ we will construct a computation $g(\phi) \in ACC(\cA, t_c)$, as following. Let $q_2 := \phi(\epsilon)$. By the definition of $R'$, there is $(q_1, q_2) \in R'$ such that $t' \in L(\cA_{q_1})$. Since $t' \in L(\cA_{Q_{\exists a_m \wedge t_c}})$ we have $t_c \in L(\cA_{q_1})$, and therefore there is a computation $\phi_c \in ACC(\cA_{q_1}, t_c)$. Let $q_i \in Q_I$ such that $(q_i, c, q_1, q_2) \in \delta$. By defining $g(\phi) := Tree(q_i, \phi_c, \phi)$ we obtain that $g(\phi) \in ACC(\cA, t_c)$, as requested.
	
	Let $\Phi := \{g(\phi) \mid \phi \in ACC(\cA_{R'[2]}, t_c)\}$. $g(\phi)_{\geq r} = \phi$ and therefore $g$ is injective, and we conclude that $|{\Phi}| = |ACC(\cA_{R'[2]}, t_c)| \geq m-1$.
	
	We now need to find an additional computation $\phi \in ACC(\cA, t_c)$ such that $\phi \notin \Phi$, resulting $|ACC(\cA, t_c)| \geq m$.
	
	Let $Q_{\exists a_1 \wedge \dots \wedge \exists a_{m-1}} := \{q \in R[2] \mid L(\cA_q) \cap L_{\exists a_1 \wedge \dots \wedge \exists a_{m-1}} \neq \emptyset \}$ and let
	$Q_{t_c \wedge \exists a_1 \wedge \dots \wedge \exists a_{m-1}} := \{q \in Q_{\exists a_1 \wedge \dots \wedge \exists a_{m-1}} \mid t_c \in L(\cA_q) \}$ and $Q_{\neg t_c \wedge \exists a_1 \wedge \dots \wedge \exists a_{m-1}} := Q_{\exists a_1 \wedge \dots \wedge \exists a_{m-1}} \setminus Q_{t_c \wedge \exists a_1 \wedge \dots \wedge \exists a_{m-1}}$.
	\begin{claim} \label{claim-exists-tree-in-Q-a1-am}
		There is a full-binary tree $t'' \in L_{\exists a_1 \wedge \dots \wedge \exists a_{m-1}} \cap T^\omega_{\{c, a_1, \dots, a_{m-1}\}}$ such that $t'' \in L(\cA_{Q_{t_c \wedge \exists a_1 \wedge \dots \wedge \exists a_{m-1}}})$ and $t'' \notin L(\cA_{Q_{\neg t_c \wedge \exists a_1 \wedge \dots \wedge \exists a_{m-1}}})$.
	\end{claim}
	\begin{claimproof}
		By the definition of $R[2]$ we have $L_{\exists a_1 \wedge \dots \wedge \exists a_{m-1}} \cap T^\omega_{\{c, a_1, \dots, a_{m-1}\}} \subseteq L(\cA_{R[2]})$ and therefore by the definition of $Q_{t_c \wedge \exists a_1 \wedge \dots \wedge \exists a_{m-1}}$ and $Q_{\neg t_c \wedge \exists a_1 \wedge \dots \wedge \exists a_{m-1}}$, we have $L_{\exists a_1 \wedge \dots \wedge \exists a_{m-1}} \cap T^\omega_{\{c, a_1, \dots, a_{m-1}\}} \subseteq L(\cA_{Q_{t_c \wedge \exists a_1 \wedge \dots \wedge \exists a_{m-1}}}) \cup L(\cA_{Q_{\neg t_c \wedge \exists a_1 \wedge \dots \wedge \exists a_{m-1}}})$.
		
		Assume, for the sake of contradiction, that the claim does not hold. Then, we obtain $L_{\exists a_1 \wedge \dots \wedge \exists a_{m-1}} \cap T^\omega_{\{c, a_1, \dots, a_{m-1}\}} \subseteq L(\cA_{Q_{\neg t_c \wedge \exists a_1 \wedge \dots \wedge \exists a_{m-1}}})$. We have $t_c \notin L(\cA_{Q_{\neg t_c \wedge \exists a_1 \wedge \dots \wedge \exists a_{m-1}}})$, and therefore by Lemma~\ref{lemma-not-finite} we conclude that $L(\cA_{Q_{\neg t_c \wedge \exists a_1 \wedge \dots \wedge \exists a_{m-1}}})$ is not finitely ambiguous~-- a contradiction to $\cA$ being finitely ambiguous.
	\end{claimproof}
	Let $t''$ be a tree as in Claim \ref{claim-exists-tree-in-Q-a1-am}. We have $t'' \in L_{\exists a_1 \wedge \dots \wedge \exists a_{m-1}} \cap T^\omega_{\{c, a_1, \dots, a_{m-1}\}}$, and therefore $Tree(c, t_c, t'') \in L_{\neg a_1 \vee \dots \vee \neg a_m} = L(\cA)$, and there is a computation $\phi \in ACC(\cA, Tree(c, t_c, t''))$. Let $q := \phi(r)$. By definition of $t''$, we have $q \in Q_{t_c \wedge \exists a_1 \wedge \dots \wedge \exists a_{m-1}}$ and therefore $t_c \in L(\cA_q)$. Let $\phi_c \in ACC(\cA_q, t_c)$, and let $\phi'$ be the computation obtained from $\phi$ by grafting $\phi_c$ on $r$. We conclude that $\phi' \in ACC(\cA, t_c)$.
	
	Assume, for the sake of contradiction, that $\phi' \in {\Phi}$, and let $q_1 := \phi'(l)$ and $q_2 := \phi'(r)$. We have $t' \in L(\cA_{q_1})$ (by definition of $|\Phi|$) and $t'' \in L(\cA_{q_2})$ (by definition of $\phi'$). Therefore, by grafting computations $\phi_{t'} \in ACC(\cA_{q_1}, t')$ and $\phi_{t''} \in ACC(\cA_{q_2}, t'')$ to the left and right children of the root of $t_c$, respectively, we obtain $Tree(c, t', t'') \in L(\cA)$. That is a contradiction, since $t'$ contains an $a_m$ labeled node, and $t''$ contains $a_1, \dots, a_{m-1}$ labeled nodes, and therefore $Tree(c, t', t'') \notin L_{\neg a_1 \vee \dots \vee \neg a_m}$.
	
	We conclude that $\phi' \notin {\Phi}$, and therefore $|ACC(\cA, t_c)| \geq 1 + |{\Phi}| = 1 + (m - 1) = m$.
\end{proof}

\section{Finitely Ambiguous Languages} \label{sect:finite}
\begin{defi}
	Let $\Sigma = \{a_1, a_2, c\}$. We define the following languages over $\Sigma$:
	\begin{itemize}
		\item For $k,m \in \nat$ such that $k<m$, we define $L_{k,m}$ as the set of trees that are obtained from $t_c$ by grafting a tree $t' \in L_{\neg a_1 \vee \neg a_2}$ on node $l^k r$, and grafting $t_{a_1}$ on node $l^m$.
		\item For $m \in \nat$ we define $L_m := \cup_{k<m} L_{k,m}$.
		\item $L^{fa} := \cup_{m \in \nat} L_m$.
	\end{itemize}
\end{defi}
In this section we prove the following proposition:
\begin{prop} The degree of ambiguity of $L^{fa}$ is finite.
\end{prop}
The proposition follows from Lemma~\ref{lemma-accepted-by-finitiely-ambiguous-automata} and Lemma~\ref{lemma-at-least-finitely-ambiguous} proved below.

\begin{lem} \label{lemma-accepted-by-finitiely-ambiguous-automata}
	There is a finitely ambiguous automaton that accepts $L^{fa}$
\end{lem}
\begin{proof}
On a tree $t\in L_m$  the automaton ``guesses''  a position $i<m$, checks that $t_{\geq l^ir} \in  L_{\neg a_1 \vee \neg a_2}$ (using  a 2-ambiguous automaton),
	checks that $t_{\geq l^jr} =t_c$  for all  $j\neq i\wedge j<m$, and checks that $t_{\geq l^m} =t_{a_1}$ (using deterministic automata).
	Below, a more detailed proof is given.

	First, notice that there are deterministic PTA $\cA_c$, $\cA_{a_1}$, $\cA_{\neg a_1}$ and $\cA_{\neg a_2}$ that accepts languages $\{t_c\}$, $\{t_{a_1}\}$, $L_{\neg a_1}$ and $L_{\neg a_2}$, respectively.
	
	By Lemma~\ref{lemma-ambiguity-of-union-and-intersection}, there is a $2$-ambiguous automaton $\cA_{\neg a_1 \vee \neg a_2}$ that accepts the language $L_{\neg a_1 \vee \neg a_2} := L_{\neg a_1} \cup L_{\neg a_2}$.
	
	We will construct an automaton $\cB := (Q_\cB,\Sigma_\cB, Q_{I_\cB}, \delta_\cB, \mathbb{C}_\cB)$ that accepts $L^{fa}$.
	
	\begin{itemize}
		\item $Q_\cB$ is defined as the union of states of $\cA_{a_1}$, $\cA_c$ and $\cA_{\neg a_1 \vee \neg a_2}$, along with additional states $q_1, q_2$.
		\item $\Sigma_\cB := \{a_1, a_2, c\}$
		\item $Q_{I_\cB} := \{q_1\}$
		\item $\delta_\cB$ will consists of the transitions of $\cA_{a_1}$, $\cA_c$ and $\cA_{\neg a_1 \vee \neg a_2}$, along with additional transitions:
		\begin{itemize}
			\item $(q_1, c, q_1, p) \in \delta_\cB$ for $p$ an initial state in $\cA_c$
			\item $(q_1, c, q_2, p) \in \delta_\cB$ for $p$ an initial state in $\cA_{\neg a_1 \vee \neg a_2}$
			\item $(q_2, c, q_2, p) \in \delta_\cB$ for $p$ an initial state in $\cA_c$
			\item $(q_2, a_1, p, p) \in \delta_\cB$ for $p$ an initial state in $\cA_{a_1}$
		\end{itemize}
		\item $\mathbb{C}_\cB(q_1) := 1$, $\mathbb{C}_\cB(q_2) := 1$, and for other states, the assigned color would be the same as in the automaton the state has originated from ($\cA_{a_1}$, $\cA_c$ or $\cA_{\neg a_1 \vee \neg a_2}$)
	\end{itemize}
	
	It is easy to see that $L(\cB) = L^{fa}$.
	
	Let $t \in L(\cB)$. By definition of $L^{fa}$, there is $m \in \nat$ such that $t \in L_m$.
If $\phi$ is an accepting computation  on $t$, then $\phi$ assigns to
the first $m+2$ nodes on the leftmost branch    the sequence $\underbrace{q_1, \dots, q_1}_\text{$i$ times} \cdot \underbrace{q_2, \dots, q_2}_\text{$m-i+1$ times} \cdot q_{a_1}$ for some $i\in \{1,\dots ,m\}$, where $q_{a_1}$ is the initial state of  $\cA_{a_1}$  (total $m$ possibilities).  $\phi$ assigns  to  $l^j \cdot r$ the initial state of $\cA_c$ if $j<i-1$ or $i-1<j<m$; the initial state of $\cA_{\neg a_1 \vee \neg a_2}$ if $j=i-1$; and the initial state of  $\cA_{a_1}$ if $j\geq m$. Since $\cA_c$  and $\cA_{a_1}$ are deterministic and $\cA_{\neg a_1 \vee \neg a_2}$ is 2-ambiguous,
the number of accepting computations on $t$  is at most  $2m $, hence, finite.
 %
%
%
\end{proof}

\begin{lem} \label{lemma-L-at-least-m-1-ambiguous}
	Let $L$ be a tree language such that $L_m \subseteq L \subseteq L^{fa}$. Then, $L$ is not $m-1$ ambiguous.
\end{lem}
\begin{proof}
	Let $\cA$ be an automaton with states $Q$ that accepts $L$, and assume $\cA$ is finitely ambiguous.
	Define a set $Q' \subseteq Q$ by $Q' := \{\phi(l^i r) \mid i<m \wedge \exists t \in L: \phi \in ACC(\cA, t) \}$
	and $Q_{\exists a_1} := \{q \in Q' \mid L_{\exists a_1} \cap L(\cA_q) \neq \emptyset \}$, and let
	$Q_{t_c \wedge \exists a_1} := \{q \in Q_{\exists a_1} \mid t_c \in L(\cA_q) \}$ and
	$Q_{\neg t_c \wedge \exists a_1} := Q_{\exists a_1} \setminus Q_{t_c \wedge \exists a_1}$.

	Relying on the fact that $ T^\omega_{\{c,a_1\}} \setminus \{t_c\}$ is not finitely ambiguous (by Corollary~\ref{corollary-not-finitely-ambiguous-languages}~(\ref{corollary-singleton-complement})), we derive the following claim:
	\begin{claim} \label{claim-Q_tc_and_exists_-has-a-tree}
		There is a tree $t_{\exists a_1} \in \big( T^\omega_{\{c,a_1\}} \setminus \{t_c\}\big)\cap \big(L(\cA_{Q_{t_c \wedge \exists a_1}})\setminus L(\cA_{Q_{\neg t_c \wedge \exists a_1}})\big)$.
		\xqed{\blacksquare}
	\end{claim}
	\noindent	Recall that $t^m$ is the tree that is obtained from $t_c$ by grafting $t_{a_1}$ on node $l^m$. For each $i < m$, define $t^m_i$ as the tree that is obtained from $t^m$ by grafting $t_{\exists a_1}$ on node $l^i r$. It is clear that $t^m_i \in L(\cA)$, and therefore there is an accepting computation $\phi_i$ of $\cA$ on $t^m_i$.
	
	$t_{\exists a_1} \in L(\cA_{Q_{t_c \wedge \exists a_1}}) \setminus L(\cA_{Q_{\neg t_c \wedge \exists a_1}})$ and since $t_{\exists a_1} \in \cA_{\phi_i(l^i r)}$ we conclude that $\phi_i(l^i r) \in Q_{t_c \wedge \exists a_1}$ and therefore $t_c \in L(\cA_{\phi_i(l^i r)})$. Let $\phi^c_i \in ACC(\cA_{\phi_i(l^i r)}, t_c)$, and construct a computation $\phi_i'$ from $\phi_i$ by grafting $\phi^c_i$ on $l^i r$.
	This tree that is obtained from $t^m_i$ by grafting $t_c$ on $l^i r$ is the tree $t^m$ and therefore $\phi_i' \in ACC(\cA, t^m)$.
	
	We are going to show that for all $i < j < m$, the computations $\phi_i', \phi_j' \in ACC(\cA, t^m)$ are different. Assume towards a contradiction $\phi_i' = \phi_j'$ and let $\widehat{\phi} := \phi_i'$. Define $p_i := \widehat{\phi}(l^i r)$, $p_j := \widehat{\phi}(l^j r)$, and let $\phi_{p_i} \in ACC(\cA_{p_i}, t_{\exists a_1})$ and $\phi_{p_j} \in ACC(\cA_{p_2}, t_{\exists a_1})$. Construct $t'$ from $t^m$ by grafting $t_{\exists a_1}$ on nodes $l^i r$ and $l^j r$, and construct $\phi'$ from $\widehat{\phi}$ by grafting $\phi_{p_i}$ on $l^i r$ and $\phi_{p_2}$ on $l^j r$. It follows that $\phi'$ is an accepting computation of $\cA$ on $t'$. That is a contradiction, since $t' \notin L^{fa}$ (since $t'_{\geq l^j r}=t'_{\geq l^i r} =t_{\exists a_1} \neq t_c$) and therefore $t' \notin L$ (since $L \subseteq L^{fa}$). We conclude that there are at least $m$ different accepting computations of $\cA$ on $t^m$.
\end{proof}
\begin{rem}
	The language $L_m$ is $2m$ ambiguous but not $m-1$ ambiguous.
 This implies that the hierarchy of ambiguous languages is infinite.
	The point of the more complex construction in Section~\ref{sect:k-ambig} is to show that this hierarchy is populated at every level.
\end{rem}
\begin{lem} \label{lemma-at-least-finitely-ambiguous}
	$L^{fa}$ is not boundedly ambiguous
\end{lem}
\begin{proof}
	$\forall m \in \nat: L_m \subseteq L^{fa}$, and therefore from Lemma~\ref{lemma-L-at-least-m-1-ambiguous} it follows that $L^{fa}$ is not $(m-1)$-ambiguous. That is, $L^{fa}$ is not boundedly ambiguous.
\end{proof}

%

\section{Uncountably Ambiguous Languages} \label{sect:uncount}
In this section we introduce a scheme for obtaining uncountably ambiguous languages from languages that are not boundedly  ambiguous. We then use this scheme to obtain natural examples of uncountably  ambiguous tree languages.

\begin{defi}
	Let $L^{\neg ba}$ be an arbitrary regular tree language over alphabet $\Sigma$ that is not boundedly ambiguous, and let $L_0$ be an arbitrary regular tree language over alphabet $\Sigma$ such that $L_0 \cap L^{\neg ba} = \emptyset$.
	Let $c \in \Sigma$ and define a language $\langUncountable$ over alphabet $\Sigma$: $t \in \langUncountable$ iff the following conditions hold:
	\begin{itemize}
		\item $\forall v \in l^*: t(v) = c$
		\item There is an infinite set $I \subseteq \nat$ such that $\forall i \in I: t_{\geq l^i \cdot r} \in L^{\neg ba}$
		and $\forall i \not \in I: t_{\geq l^i \cdot r} \in L_0$.
	\end{itemize}
\end{defi}

\begin{prop} \label{lemma-uncountably-ambiguous-language}
	The degree of ambiguity of $\langUncountable$ is $2^{\aleph_0}$.
\end{prop}
\begin{proof}
	Let $\cA = (Q,\Sigma, Q_I, \delta, \mathbb{C})$ be a PTA that accepts $\langUncountable$. We will show that $da(\cA) = 2^{\aleph_0}$.
	
        Let $Q' := \{\phi(u) \mid u \in l^* \cdot r$ and $\exists t: \phi \in ACC(\cA, t)\}$, and define $Q_{\mathit{unamb}\wedge \neg L_0} := \{q \in Q' \mid \cA_q$ is unambiguous and $L(\cA_q) \cap L_0 = \emptyset\}$.

	\begin{claim} \label{claim-contained-in-L-neg-ba}
          $L(\cA_{Q_{\mathit{unamb}\wedge \neg L_0}}) \subseteq L^{\neg ba}$.
	\end{claim}
	\begin{claimproof}
          Assume, for the sake of contradiction, that there is a tree $t \in L(\cA_{ Q_{\mathit{unamb}\wedge \neg L_0}})$ such that $t \notin L^{\neg ba}$. By definition of $ Q_{\mathit{unamb}\wedge \neg L_0}$ we conclude that $t \notin L_0$.
		
                Let $q \in Q_{\mathit{unamb}\wedge \neg L_0}$ such that $t \in L(\cA_q)$ and let $\phi \in ACC(\cA_q, t)$. Recall that $q \in Q'$ (since $Q_{\mathit{unamb}\wedge \neg L_0} \subseteq Q'$) and therefore there is a tree $t' \in L(\cA)$, a computation $\phi' \in ACC(\cA, t')$ and a node $u \in l^* \cdot r$ such that $\phi'(u) = q$. By the grafting lemma we conclude that $\phi' [ \phi/u]$ is an accepting computation of $\cA$ on $t'[t/u]$. Therefore, $t'  [t/u] \in L(\cA)$ for $t \notin L^{\neg ba} \cup L_0$~-- a contradiction to the definition of $\cA$.
	\end{claimproof}
	
        Notice that $L(\cA_{Q_{\mathit{unamb}\wedge \neg L_0}})$ is boundedly ambiguous by Corollary~\ref{corollary-ambiguity-closure} (as a finite union of unambiguous languages), and since $L^{\neg ba}$ is not boundedly ambiguous we conclude that $da(L(\cA_{Q_{\mathit{unamb}\wedge \neg L_0}})) \neq da(L^{\neg ba})$. By Claim \ref{claim-contained-in-L-neg-ba} we obtain $L(\cA_{Q_{\mathit{unamb}\wedge \neg L_0}}) \subseteq L^{\neg ba}$, and applying Lemma~\ref{lemma-subset-language-of-different-ambiguity} we conclude that there is a tree $t_{\neg ba} \in L^{\neg ba}$ such that $t_{\neg ba} \notin L(\cA_{Q_{\mathit{unamb}\wedge \neg L_0}})$.
	
	Let $c \in \Sigma$ be as in the definition of $\langUncountable$, and let $t_c$ be a tree where all nodes are labeled by $c$. Let $A := l^* \cdot r$ be an antichain, and define $t'' := t_c [t_{\neg ba}/A]$. By the definition of $\cA$ it is clear that $t'' \in L(\cA)$. Let $\phi'' \in ACC(\cA, t'')$, and let $B := \{u \in A \mid L(\cA_{\phi''(u)}) \cap L_0 \neq \emptyset \}$.
	
	For each $u \in B$ there is a tree $t_u \in L_0$ and a computation $\phi_u \in ACC(\cA_{\phi''(u)}, t_u)$. Therefore, by the grafting lemma, we conclude that the tree $t'''$ that is obtained from $t''$ by grafting $t_u$ on each node $u \in B$ is in $L(\cA)$.
	
	Assume, for the sake of contradiction, that $A \setminus B$ is finite. By definition of $t'''$, for each $i \in \nat$ such that $u := l^i \cdot r \in B$ we have $t'''_{\geq l^i \cdot r} = t_u \in L_0$. Therefore, $|\{i \in \nat \mid t'''_{\geq l^i \cdot r} \in L^{\neg ba}\}| = |\{u \in A \mid t'''_{\geq u} \in L^{\neg ba}\}| = |\{u \in A \setminus B \mid t'''_{\geq u} \in L^{\neg ba}\}| = |A \setminus B| < \aleph_0$, and by definition of $\langUncountable$ we conclude that $t''' \notin \langUncountable$~-- a contradiction to the definition of $\cA$.
	
        $A \setminus B$ is infinite, and therefore there is a state $q$ and an infinite set $\widehat{A} \subseteq A \setminus B$ such that $\phi''(u) = q$ for all $u \in \widehat{A}$. Recall that $\forall u \in \widehat{A}: t''_{\geq u} = t_{\neg ba}$. Notice that for each $u \in \widehat{A}$ we have $u \notin B$, and by definition of $B$ we obtain $L(\cA_{\phi''(u)}) \cap L_0 = L(\cA_q) \cap L_0 = \emptyset$. Since $t_{\neg ba} \notin L(\cA_{Q_{\mathit{unamb}\wedge \neg L_0}})$ we conclude that $q \notin Q_{\mathit{unamb}\wedge \neg L_0}$~-- hence, $\cA_q$ is ambiguous.

        Let $t_{\mathit{amb}} \in L(\cA_q)$ be a tree having at least two accepting computations $\phi_1, \phi_2 \in ACC(\cA_q, t_{\mathit{amb}})$.
        Let $\widehat{t} := t'' [ t_{\mathit{amb}}/{\widehat{A}}]$, and $\widehat{\phi} := \phi [ \phi_1/{\widehat{A}}]$. By the grafting lemma we obtain $\widehat{\phi} \in ACC(\cA, \widehat{t})$. For each $A' \subseteq \widehat{A}$, define a computation $\phi_{A'} := \widehat{\phi} [ \phi_2/{A'} ]$. Notice that $\phi_{A'} \in ACC(\cA, \widehat{t})$ (by the grafting lemma) and that $\forall A_1, A_2 \subseteq \widehat{A}: A_1 \neq A_2 \rightarrow \phi_{A_1} \neq \phi_{A_2}$ (since $\phi_1 \neq \phi_2$). Therefore, $|ACC(\cA, \widehat{t})| \geq |\{A' \mid A' \subseteq \widehat{A}\}| = 2^{\aleph_0}$, and $da(\cA) = 2^{\aleph_0}$, as requested.
\end{proof}

We will now introduce a couple of definitions, and present three  natural examples of infinite tree languages that are not countably ambiguous.

\begin{defi}[Characteristic tree]
	The characteristic tree of $U_1, \dots, U_n \subseteq \directions^*$ is a $\{0,1\}^n$-labeled tree $t[U_1, \dots, U_n]$ such that $t[U_1, \dots, U_n](u) := (b_1, \dots, b_n)$ where $b_i = 1$ iff $u \in U_i$ for each $1 \leq i \leq n$.
\end{defi}

\begin{defi}
	For a set $U \subseteq \directions^*$ we define $U\downarrow$ as the downward closure of $U$.
\end{defi}

\begin{defi}
	A set $X \subseteq \directions^*$ is called \textbf{perfect} if $X \neq \emptyset$ and $\forall u \in X: \exists v_1, v_2 \in X$ such that $v_1$ and $  v_2 $ are incomparable and greater than $u$.  
\end{defi}

\begin{prop} \label{prop-examples-of-uncountably-ambiguous-lanauges}
	The following regular languages are not countably ambiguous:
	\begin{enumerate}
		\item $L_{X \subseteq Y\downarrow} := \{t[X, Y] \mid X \subseteq Y\downarrow\}$~-- ``for each node in $X$ there is a greater or equal node in $Y$.''
                \item $L_{\mathit{no-max}} := \{t[X] \mid X$ has no maximal element$\}$~-- ``for each node in $X$ there is a greater node in $X$.''
	\item $L_{\perf} := \{t[X] \mid X$ is perfect $\}$~-- ``for each node in $X$ there are at least two greater incomparable nodes in $X$.''	
	\end{enumerate}
\end{prop}
In the rest of this section we will prove Proposition \ref{prop-examples-of-uncountably-ambiguous-lanauges}.

\begin{proof}[Proof of Proposition \ref{prop-examples-of-uncountably-ambiguous-lanauges}(1)]
  Let $L_{\mathit{left}} := \{t[X,Y] \mid X = l^*$ and $Y \cap l^* = \emptyset\}$. It is easy to see that $L_{\mathit{left}}$ can be accepted by a deterministic PTA, and therefore $da(L_{\mathit{left}}) = 1$.
	
        By Lemma~\ref{lemma-ambiguity-of-union-and-intersection} we conclude that $da(L_{X \subseteq Y\downarrow} \cap L_{\mathit{left}}) \leq da(L_{X \subseteq Y\downarrow}) \cdot da(L_{\mathit{left}}) = da(L_{X \subseteq Y\downarrow})$. We will show that $L_{X \subseteq Y\downarrow} \cap L_{\mathit{left}}$ is not countably ambiguous. By the above inequality, this implies that $L_{X \subseteq Y\downarrow}$ is not countably ambiguous.
		
	\begin{claim} \label{claim-downward-closed-language-scheme}
		Let $L_{X = \emptyset, Y \neq \emptyset} := \{t[X,Y] \mid X = \emptyset$ and $Y \neq \emptyset\}$.
                Then $t' \in L_{X \subseteq Y\downarrow} \cap L_{\mathit{left}}$ iff the following conditions hold:
		\begin{enumerate}
			\item $\forall u \in l^*: t'(u) = (1,0)$
			\item There is an infinite set $I \subseteq \nat$ such that:
			\begin{enumerate}
				\item If $i \in I$ then $t'_{\geq l^i \cdot r} \in L_{X = \emptyset, Y \neq \emptyset}$
				\item If $i \notin I$ then $t'_{\geq l^i \cdot r} \in \{t[\emptyset, \emptyset]\}$
			\end{enumerate}
		\end{enumerate}
	\end{claim}
	\begin{claimproof}
          $\Rightarrow$: Let $t' \in L_{X \subseteq Y\downarrow} \cap L_{\mathit{left}}$. By definition of $L_{\mathit{left}}$ it is clear that the condition (1) holds, and that for each $i \in \nat: t'_{\geq v^i \cdot l} \in L_{X = \emptyset, Y \neq \emptyset}$ or $t'_{\geq v^i \cdot l} = t[\emptyset, \emptyset]$. Assume, for the sake of contradiction, that the set $\{i \in \nat \mid t'_{\geq v^i \cdot l} \in L_{X = \emptyset, Y \neq \emptyset}\}$ is finite. Therefore, by the second condition, there is an index $k \in \nat$ such that $\forall i \geq k: t'_{\geq v^i \cdot l} = t[\emptyset, \emptyset]$. Let $u := l^k$. By the definition of $L_{\mathit{left}}$ we have $u \in X$, and for each $v \geq u$ we have either $t'(v) = (1,0)$ if $v \in l^*$, or $t'(v) = (0,0)$ otherwise. Hence, $\forall v \geq u: v \notin Y$, in contradiction to $t' \in L_{X \subseteq Y\downarrow}$.
		
                $\Leftarrow$: Assume that the conditions hold for $t'$. It is easy to see that $t' \in L_{\mathit{left}}$. We will show that $t' \in L_{X \subseteq Y\downarrow}$. Assume, for the sake of contradiction, that there is a node $u \in X$ such that $v \notin Y$ for each node $v \geq u$. Since all nodes in $X$ are in $l^*$ we conclude that there is $i \in \nat$ such that $u = l^i$. Notice that the set $I \subseteq \nat$ is infinite, and therefore there is $j > i$ such that $t'_{\geq l^j \cdot r} \in L_{X = \emptyset, Y \neq \emptyset}$. Therefore, there is a node $v \geq l^j \cdot r > l^i = u$ such that $v \in Y$~-- a contradiction.
	\end{claimproof}
	
	Observe that the language $L_{X = \emptyset, Y \neq \emptyset} := \{t[X,Y] \mid X = \emptyset$ and $Y \neq \emptyset\}$ can be considered as a tree language over alphabet $\{0\} \times \{0,1\}$, and that $L_{X = \emptyset, Y \neq \emptyset} = T^\omega_{\{0\} \times \{0,1\}} \setminus \{t[\emptyset, \emptyset]\}$. Therefore, by Corollary~\ref{corollary-not-finitely-ambiguous-languages}(\ref{corollary-singleton-complement}) we conclude that $L_{X = \emptyset, Y \neq \emptyset}$ is not finitely ambiguous.
	
        Notice that by Claim \ref{claim-downward-closed-language-scheme} we obtain $L_{X \subseteq Y\downarrow} \cap L_{\mathit{left}} = \langUncountable$, for $L_0 = \{t[\emptyset, \emptyset]\}$ and $L^{\neg ba} = L_{X = \emptyset, Y \neq \emptyset}$. Therefore, applying Proposition \ref{lemma-uncountably-ambiguous-language} we conclude that $L_{X \subseteq Y\downarrow} \cap L_{\mathit{left}}$ is not countably ambiguous.
\end{proof}

To prove Proposition \ref{prop-examples-of-uncountably-ambiguous-lanauges}(2), we will first prove the following lemma:

\begin{lem} \label{lemma-no-max-not-finitely-amb}
  $L_{\mathit{no-max}}$ is not finitely ambiguous.
\end{lem}
\begin{proof}
  Let $\cA = (Q, \Sigma, Q_I, \delta, \mathbb{C})$ be a PTA that accepts $L_{\mathit{no-max}}$. Let $Q' := \{q \in Q \mid \exists q_i \in Q_I \exists q' \in Q: (q_i, 1, q, q') \in \delta$ and $t[\emptyset] \in L(\cA_{q'})\}$.
			
	\begin{claim} \label{claim-no-max-properties}
		Define $L_{\neg \emptyset} := T^\omega_\Sigma \setminus \{t[\emptyset]\}$. Then:
		\begin{enumerate}
                  \item $L_{\mathit{no-max}} \setminus \{t[\emptyset]\} \subseteq L(\cA_{Q'})$
			\item $L(\cA_{Q'}) \subseteq L_{\neg \emptyset}$
		\end{enumerate}
	\end{claim}
	\begin{claimproof}
          (1) Let $t' \in L_{\mathit{no-max}} \setminus \{t[\emptyset]\}$, and let $t_\epsilon := t[\{\epsilon\}]$ (that is, $t_\epsilon(\epsilon) := 1$, and $\forall u \neq \epsilon: t_\epsilon(u) := 0$).		
                Let $t'' := t_\epsilon [t'/l ]$. By the definition of $L_{\mathit{no-max}}$ we obtain $t'' \in L_{\mathit{no-max}}$. Therefore, there is a computation $\phi \in ACC(\cA, t'')$ such that $\phi(l) \in Q'$ and $t' \in L(\cA_{\phi(l)})$, as requested.		
		
                (2) Assume, for the sake of contradiction, that $t[\emptyset] \in L(\cA_{Q'})$. Then there is a transition $(q_i, 1, q_1, q_2) \in \delta$ from an initial state $q_i$ such that $t[\emptyset] \in L(\cA_{q_1})$ and $t[\emptyset] \in L(\cA_{q_2})$. Therefore, we conclude that $t_\epsilon := t[\{\epsilon\}]$ is accepted by $\cA$~-- a contradiction to the definition of~$L_{\mathit{no-max}}$.
	\end{claimproof}
	
	Let $\Sigma:=\{0,1\}$. Define a function $F: \Sigma^* \rightarrow \Sigma$ such that \[F(\sigma_1, \dots, \sigma_m) := \begin{cases}
		1 & \exists 1 \leq i \leq m: \sigma_i = 1\\
		0 & \text{otherwise}
	\end{cases}\]
	
	It is easy to see that $F$ is definable by a Moore machine. We show that $F$ reduces $L_{\neg \emptyset}$ to $L(\cA_{Q'})$.

  Notice that $\forall t' \in T^\omega_\Sigma: t' \in L_{\neg \emptyset} \rightarrow \widehat{F}(t') \in L_{\mathit{no-max}} \setminus \{t[\emptyset]\}$. Since $L_{\mathit{no-max}} \setminus \{t[\emptyset]\} \subseteq L(\cA_{Q'})$ (by Claim \ref{claim-no-max-properties}(1)) we conclude that $\forall t' \in T^\omega_\Sigma: t' \in L_{\neg \emptyset} \rightarrow \widehat{F}(t') \in L(\cA_{Q'})$. Conversely, $\forall t' \in T^\omega_\Sigma: \widehat{F}(t') \in L_{\neg \emptyset} \rightarrow t' \in L_{\neg \emptyset}$, and since $L(\cA_{Q'}) \subseteq L_{\neg \emptyset}$ (by Claim \ref{claim-no-max-properties}(2)) we obtain $\forall t' \in T^\omega_\Sigma: \widehat{F}(t') \in L(\cA_{Q'}) \rightarrow t' \in L_{\neg \emptyset}$.
	
	Therefore, by Lemma~\ref{lemma-F-reduction}, we conclude that $da(L(\cA_{Q'})) \geq da(L_{\neg \emptyset})$. Notice that $L_{\neg \emptyset} = T^\omega_\Sigma \setminus \{t[\emptyset]\}$ and by Corollary~\ref{corollary-not-finitely-ambiguous-languages}(\ref{corollary-singleton-complement}) we obtain $da(L_{\neg \emptyset}) \geq \aleph_0$. Hence, $\cA_{Q'}$ is not finitely ambiguous, and by Corollary~\ref{corollary-A_Q-ambiguity} we conclude that $da(\cA) \geq \aleph_0$.
\end{proof}

\begin{proof}[Proof of Proposition \ref{prop-examples-of-uncountably-ambiguous-lanauges}(2)]
	Let $L_{l^* \cap X = \emptyset} := \{t[X] \mid X \cap l^* = \emptyset\}$. It is easy to construct a deterministic PTA that accepts $L_{l^* \cap X = \emptyset}$, and therefore $da(L_{l^* \cap X = \emptyset}) = 1$.
	
	By Lemma~\ref{lemma-ambiguity-of-union-and-intersection} we conclude that
        $da(L_{\mathit{no-max}} \cap L_{l^* \cap X = \emptyset}) \leq da(L_{\mathit{no-max}}) \cdot da(L_{l^* \cap X = \emptyset}) = da(L_{\mathit{no-max}})$. We will show that $da(L_{\mathit{no-max}} \cap L_{l^* \cap X = \emptyset}) = 2^{\aleph_0}$, and the lemma will follow.	
	
        Notice that $t' \in L_{\mathit{no-max}} \cap L_{l^* \cap X = \emptyset}$ iff the following hold:
	\begin{itemize}
		\item $\forall u \in l^*: t(u) = 0$
                \item $\forall u \in l^* \cdot r: t'_{\geq u} \in L_{\mathit{no-max}}$
	\end{itemize}
	
        It is easy to see that $L_{\mathit{no-max}} \cap L_{l^* \cap X = \emptyset} = \langUncountable$ for $L^{\neg ba} := L_{\mathit{no-max}}$ (which is not boundedly ambiguous, by Lemma~\ref{lemma-no-max-not-finitely-amb}) and $L_0 := \emptyset$. Therefore, by Proposition \ref{lemma-uncountably-ambiguous-language} we conclude that $da(L_{\mathit{no-max}} \cap L_{l^* \cap X = \emptyset}) = 2^{\aleph_0}$, as requested.
\end{proof}

\begin{proof}[Proof of Proposition \ref{prop-examples-of-uncountably-ambiguous-lanauges}(3)]
  Let $L_{\mathit{contains-}l^*} := \{t[X] \mid l^* \subseteq X \}$. It is easy to see that $L_{\mathit{contains-}l^*}$ can be accepted by a deterministic PTA, and therefore $da(L_{\mathit{contains-}l^*}) = 1$.
        Look at the language $L_{\perf} \cap L_{\mathit{contains-}l^*}$. By Lemma~\ref{lemma-ambiguity-of-union-and-intersection} we obtain $da(L_{\perf} \cap L_{\mathit{contains-}l^*}) \leq da(L_{\perf}) \cdot da(L_{\mathit{contains-}l^*}) = da(L_{\perf})$. We will show that $L_{\perf} \cap L_{\mathit{contains-}l^*}$ is not countably ambiguous. By the above inequality, this implies that $da(L_{\perf}) = 2^{\aleph_0}$.

	\begin{claim} \label{claim-L-perf-not-finitely-amb}
		$L_{\perf}$ is not finitely ambiguous.
	\end{claim}
	\begin{claimproof}
		Define a function $F: \Sigma^* \rightarrow \Sigma$ such that $F(\sigma_1, \dots, \sigma_m) := \begin{cases}
		1 & \exists 1 \leq i \leq m: \sigma_i = 1\\
0 & \text{otherwise.}
		\end{cases}$. It is easy to see that $F$ is definable by a Moore machine, and that $\forall t' \in T^\omega_\Sigma: t' \in T^\omega_\Sigma \setminus \{t[\emptyset]\} \leftrightarrow \widehat{F}(t') \in L_{\perf}$. Notice that $T^\omega_\Sigma \setminus \{t[\emptyset]\}$ is not finitely ambiguous (by Corollary~\ref{corollary-not-finitely-ambiguous-languages}(\ref{corollary-singleton-complement})), and therefore by Lemma~\ref{lemma-F-reduction} we conclude that $L_{\perf}$ is not finitely ambiguous.
	\end{claimproof}

	\begin{claim}
          $t' \in L_{\perf} \cap L_{\mathit{contains-}l^*}$ iff the following conditions hold:
		\begin{enumerate}
			\item $\forall u \in l^*: t'(u) = 1$
			\item There is an infinite set $I \subseteq \nat$ such that $\forall i \in I: t'_{\geq l^i \cdot r} \in L_{\perf}$ and $\forall i \not \in I: t'_{\geq l^i \cdot r} \in \{t[\emptyset]\}$.
		\end{enumerate}	
	\end{claim}
	\begin{claimproof}
          $\Rightarrow$: Let $t' \in L_{\perf} \cap L_{\mathit{contains-}l^*}$. By definition of $L_{\mathit{contains-}l^*}$ it is clear that condition (1) holds for $t'$. Notice that $\forall i \in \nat: t'_{\geq l^i \cdot r} \in L_{\perf}$ or $t'_{\geq l^i \cdot r} = t[\emptyset]$. Assume, for the sake of contradiction, that $\{i \in \nat \mid t'_{\geq l^i \cdot r} \in L_{\perf}\}$ is finite. Therefore, there is $k \in \nat$ such that $\forall i \geq k: t'_{\geq l^i \cdot r} = t[\emptyset]$. Let $u := l^k$, and notice that $t'(u) = 1$, and $\forall v > u: t'(v) = 1 \leftrightarrow v \in l^*$. Hence, each pair of $1$-labeled nodes that are greater than $u$ are comparable~-- a contradiction to the definition of $L_{\perf}$.
		
                $\Leftarrow$: Let $t'$ such that the conditions hold. By the first condition it is clear that $t' \in L_{\mathit{contains-}l^*}$. We will prove that $t' \in L_{\perf}$, and the claim will follow. First, notice that $t'(\epsilon) = 1$, and therefore $t' \neq t[\emptyset]$. Let $u$ be a node such that $t'(u) = 1$. If $u \in l^*$ then by the second condition, there is a node $v \in l^* \cdot r$ such that $v > u$ and $t_{\geq v} \in L_{\perf}$. Therefore, there are two nodes $w_1, w_2 > v > u$ such that $w_1 \perp w_2$ and $t'(w_1) = t'(w_2) = 1$. Otherwise ($u \notin l^*$), there is a node $v \in l^* \cdot r$, such that $u > v$ and $t_{\geq v} \in L_{\perf}$, and by definition of $L_{\perf}$ we conclude that there are two nodes $w_1, w_2 > u$ such that $w_1 \perp w_2$ and $t'(w_1) = t'(w_2) = 1$~-- hence, $t' \in L_{\perf}$.		
	\end{claimproof}

        It is easy to see that $L_{\perf} \cap L_{\mathit{contains-}l^*} = \langUncountable$ for $L^{\neg ba} := L_{\perf}$ (which is not boundedly ambiguous, by Claim \ref{claim-L-perf-not-finitely-amb}) and $L_0 := \{t[\emptyset]\}$. Therefore, by Proposition \ref{lemma-uncountably-ambiguous-language} we conclude that $L_{\perf} \cap L_{\mathit{contains-}l^*} = 2^{\aleph_0}$, as requested.
\end{proof}
Observe that our proof shows that $L_{\perf\wedge \min} := \{t[X] \mid X$ is perfect and has the $\leq$-minimal element$\}$ is also uncountably ambiguous.
We conclude with an instructive  example of an unambiguous language which is similar to $L_{\perf\wedge \min}$.
Let $X \subseteq \{l,r\}^*$ be a set of nodes. We say that $u \in X$  is a  {$X$-successor} of $v$ if $u > v$ and there is no node $w \in X$ such that $v < w < u$.
	We call $X$ a  {full-binary subset-tree} if $X$ has a minimal node, and each node in $X$   has two $X$-successors.
	
Note that if $X$ is a full-binary subset tree then $X$ is perfect and has the $\leq$-minimal element.
However the language $L_{\mathit{binary}}:=\{t[X]\mid 	 X$ is a full-binary subset tree$\}$ is unambiguous.

\section{Countable Languages are Unambiguous} \label{sect:countable-langauges-ambiguity}
%
In this section we prove the following Proposition:
\begin{prop} \label{prop-countable-tree-languages-unambiguous}
	Each regular countable tree language is unambiguous
\end{prop}
This section is self-contained and lacks  technical connections to the previous sections.
It is instructive to compare the above Proposition with    Corollary~\ref{corollary-not-finitely-ambiguous-languages}(1)
which states that the complement of every countable tree language is not finitely ambiguous.

To prove Proposition \ref{prop-countable-tree-languages-unambiguous} we first recall finite tree automata (Subsection~\ref{subsection-finite-tree-automata}). Then, we present Niwi\'nski's Representation for Countable Languages (Subsection~\ref{subsection-countable-languages-representation}). Finally, the proof of Proposition \ref{prop-countable-tree-languages-unambiguous} is given (Subsection~\ref{subsection-proof-for-countable-languages-ambiguity}).

\subsection{Finite Trees and Finite Tree Automata} \label{subsection-finite-tree-automata}
~

\paragraph{Finite Trees}
A finite tree is a finite set $U \subseteq \directions^*$ that is closed under prefix relation. $U$ is called a finite \textbf{binary} tree if $\forall u \in U: u \cdot l \in U \leftrightarrow u \cdot r \in U$.

\paragraph{Finite $\Sigma$-labeled Binary Trees}
Let $\Sigma$ be partitioned into two sets: $\Sigma_2$~-- labels of internal nodes, and $\Sigma_0$~-- labels of leaves. A finite $\Sigma$-labeled binary tree is a function $t_U : U \rightarrow \Sigma$, where $U \subseteq \{l,r\}^*$ is a finite binary tree, $t_U(v)\in \Sigma_0$ if $v$ is a leaf, and $t_U(v)\in \Sigma_2$ if $v$ has children.

When it is clear from the context, we will use ``finite tree'' or ``labeled finite tree'' for ``$\Sigma$-labeled finite binary tree''.

\paragraph{Finite Tree Automata (FTA)}
An automaton over $\Sigma$-labeled finite trees is a tuple $\cB = (Q, \Sigma, Q_I, \delta)$, where $Q$ is a finite set of states, $\Sigma = \Sigma_0 \cup \Sigma_2$ is an alphabet, $Q_I$ is a set of initial states, and $\delta \subseteq (Q \times \Sigma_0) \cup (Q \times \Sigma_2 \times Q \times Q)$ is a  set of transitions.

An accepting computation of $\cB$ on a finite tree $t_U$ is a function $\phi: U \rightarrow Q$, such that $\phi(\epsilon) \in Q_I$, and for each node $u \in U$, if $u$ is not a leaf then $(\phi(u), t_U(u), \phi(u \cdot l), \phi(u \cdot r)) \in \delta$, and otherwise $(\phi(u), t_U(u)) \in \delta$.

The language of a FTA $\cB$ is the set of finite trees $t$ such that $\cB$ has an accepting computation on $t$.
A finite tree language is regular iff it is accepted by a FTA. It is well-known that every regular finite tree language is unambiguous (i.e., for every finite tree language there is an unambiguous automaton that accepts it).

\subsection{Niwi\'nski's Representation for Countable Languages} \label{subsection-countable-languages-representation}

\begin{defi}
	Define $\fintree$ as the set of finite trees over alphabet $\Sigma \cup \{x_1, \dots, x_n\}$ where the internal nodes are $\Sigma$-labeled, and the leaves are $\{x_1, \dots, x_n\}$-labeled.
	
	Let $\tau \in \fintree$ be a finite tree, and let $t_1, \dots, t_n \in T^\omega_\Sigma$ be infinite binary trees over alphabet $\Sigma$. We define $\tau[t_1/x_1,\dots ,t_n/x_n]$ as the infinite tree that is obtained from $\tau$ by grafting $t_i$ on leaves labeled by $x_i$.
	
	For a set $M \subseteq \fintree$, we define $M[t_1/x_1,\dots ,t_n/x_n] := \underset{\tau \in M}{\bigcup} \tau[t_1/x_1,\dots ,t_n/x_n]$.
\end{defi}

\begin{thm}[D. Niwi\'nski \cite{Damian}] \label{theorem-countable-language-representation}
	Let $L$ be a countable regular tree language over alphabet $\Sigma$. Then there is a finite set of trees $\{t_1, \dots, t_n\}$ such that the following hold:
	\begin{enumerate}
		\item For each tree $t \in L$ and a tree branch $\pi$, there is a node $v \in \pi$ and a number $1 \leq i \leq n$ such that $t_{\geq v} = t_i$.
		
		\item There is a regular finite tree language $M \subseteq \fintree$ s.t. $L = M[t_1/x_1,\dots ,t_n/x_n]$.
	\end{enumerate}
\end{thm}

The following lemma strengthen item (2) of Theorem \ref{theorem-countable-language-representation} by adding another condition on $M$, implying a unique representation of each tree in $L$:

\begin{lem} \label{lemma-unique-M}
	Let $L$ be a countable regular tree language over alphabet $\Sigma$, and let $\{t_1, \dots, t_n\}$ be a finite set of trees as in Theorem \ref{theorem-countable-language-representation}. Then there is a regular finite trees language $M \subseteq \fintree$ such that $L = M[t_1/x_1,\dots ,t_n/x_n]$, and for each $t \in L$ there is a \textbf{unique} finite tree $\tau \in M$ such that $t = \tau[t_1/x_1,\dots ,t_n/x_n]$.
\end{lem}
\begin{proof}
	For each tree $t \in L$, let $g(t)$ be the tree that is obtained from $t$ by changing the label of each node $v \in \{l,r\}^*$ where $t_{\geq v} = t_i$ to $x_i$, and removing all descendants of $\{x_1, \dots, x_n\}$-labeled node.
	
	\begin{claim}
		$g(t)$ is finite for all $t \in L$.
	\end{claim}
	\begin{claimproof}
		Assume, for the sake of contradiction, that there is $t \in L$ such that the set of nodes $U \subseteq \{l,r\}^*$ of $g(t)$ is infinite. The number of children of each node in $U$ is bounded by 2, and therefore, by K\"onig's Lemma, there is a tree branch $\pi$ such that $\forall v \in \pi: v \in U$. Therefore, by definition of $g(t)$, we conclude that $t_{\geq v} \neq t_i$ for each $v \in \pi$ and $1 \leq i \leq n$~-- a contradiction to item (1) of Theorem \ref{theorem-countable-language-representation}.
	\end{claimproof}
	
	Notice that for each $t \in L$ we obtain $g(t)[t_1/x_1,\dots ,t_n/x_n] = t$, and therefore $g$ is injective. Hence, $L = M[t_1/x_1,\dots ,t_n/x_n]$ where $M := \{g(t) \mid t \in L\}$. We will show that $M$ is a regular language of finite trees.
	
	It is easy to see that for each $t \in L$ and finite tree $\tau \in \fintree$, $\tau = g(t)$ iff the following conditions hold:
	\begin{itemize}
		\item $t = \tau[t_1/x_1,\dots ,t_n/x_n]$
		\item $t_{\geq v} \neq t_i$ for each node $v$ in $\tau$ that is not a leaf, and for each $1 \leq i \leq n$.
	\end{itemize}
	
	Since both conditions could be formulated in MSO, we conclude that $M$ is MSO-definable, and therefore regular.
\end{proof}

\subsection{Proof of Proposition \ref{prop-countable-tree-languages-unambiguous}} \label{subsection-proof-for-countable-languages-ambiguity}
Let $L$ be a countable regular tree language over alphabet~$\Sigma$. We will show that $L$ can be accepted by an unambiguous PTA.

By Lemma~\ref{lemma-unique-M}, there is a regular finite tree language $M \subseteq \fintree$ and regular infinite trees $t_1, \dots, t_n$ such that $L = M[t_1/x_1,\dots ,t_n/x_n]$. Additionally, for each $t \in L$ there is a unique $\tau \in M$ such that $t = \tau[t_1/x_1,\dots ,t_n/x_n]$.

Each infinite tree $t_i: \{l,r\}^* \rightarrow \Sigma$ is regular, and therefore, by Fact \ref{fact:sect3}, is   definable by a Moore machine $M_i = (\{l,r\}, \Sigma, Q_i, q_I^i , \delta^M_i, out^M_i)$. Let $\cA_i := (Q_i,\Sigma, q_I^i, \delta_i, F_i)$ where $F_i := Q_i$, and $(q,a,q_1,q_2) \in \delta_i$ iff $q_1 = \delta(q, l)$, $q_2 = \delta(q, r)$ and $a = out^M_i(q)$. It is easy to verify that $\cA_i$ is unambiguous, and $L(\cA_i) = \{t_i\}$.
$M$ is regular and therefore can be accepted by an unambiguous FTA $\cB = (Q_\cB,\Sigma \cup \{x_1, \dots, x_n\}, q^\cB_I, \delta_\cB)$.

We use these automata to construct a PTA $\cA := (Q,\Sigma, Q_I, \delta, \mathbb{C})$, by:
\begin{itemize}
	\item $Q := \cup_{1 \leq i \leq n} Q_i \cup Q_\cB$
	\item $q_I^i := \{q^\cB_I\} \cup \{q_I^i \mid (q^\cB_I, x_i) \in \delta_\cB \}$
	\item $\delta$ is the union of the following:
	\begin{itemize}
		\item $\{(q, a, q_1, q_2) \in \delta_\cB \mid a \in \Sigma\}$ (all transitions of $\cB$ on inner nodes)
		\item $\cup_{1 \leq i \leq n} \delta_i$
		\item $\{(q, a, q_I^i, q_I^j) \mid \exists (q, a, q_1, q_2) \in \delta_\cB: (q_1, x_i) \in \delta_\cB$ and $ (q_2, x_j) \in \delta_\cB \}$
		\item $\{(q, a, q_1, q_I^j) \mid \exists (q, a, q_1, q_2) \in \delta_\cB: (q_2, x_j) \in \delta_\cB \}$
		\item $\{(q, a, q_I^i, q_2) \mid \exists (q, a, q_1, q_2) \in \delta_\cB: (q_1, x_i) \in \delta_\cB \}$
	\end{itemize}
	\item $\mathbb{C}(q) := \begin{cases}
	\mathbb{C}_i(q) & \exists i: q \in Q_i\\
	1 & \text{otherwise}
	\end{cases}$
\end{itemize}

It is easy to see that $L(\cA) = M[t_1/x_1,\dots ,t_n/x_n] = L$.

We will show that $\cA$ is unambiguous.
For each accepting computation $\phi \in ACC(\cA, t)$, define a set of nodes $U_\phi := \{u \in \{l,r\}^* \mid \forall v <u: \phi(v) \in Q_\cB \}$. It is easy to see that $U_\phi$ is downward closed. Assume towards contradiction that $U_\phi$ is infinite~-- by K\"{o}nig Lemma, $U_\phi$ contains an infinite tree branch $\pi$. By definition of $U_\phi$ all states in $\phi(\pi)$ are in $Q_\cB$, and therefore colored by $1$. That is a contradiction to $\phi$ being an accepting computation.

Define a labeled finite tree $t_\phi: U_\phi \rightarrow \Sigma \cup \{x_1, \dots, x_n\}$ by:

$t_\phi := \begin{cases}
x_i & \exists i: \phi(u) = q_I^i \\
t(u) & \text{otherwise}
\end{cases}$

By definition of $t_\phi$ we obtain $t = t_\phi[t_1/x_1,\dots ,t_n/x_n]$, and by definition of $\cB$ we conclude that $t_\phi \in M$.

Assume, for the sake of contradiction, that $\cA$ is ambiguous. Therefore, there is a tree $t \in L$ and two distinct accepting computations $\phi_1, \phi_2 \in ACC(\cA, t)$. $\cA_i$ is deterministic for each $1 \leq i \leq n$, and therefore $\phi_1 \neq \phi_2$ iff $t_{\phi_1} \neq t_{\phi_2}$. We conclude that $t_{\phi_1}[t_1/x_1,\dots ,t_n/x_n] = t_{\phi_2}[t_1/x_1,\dots ,t_n/x_n]$ for $t_{\phi_1}, t_{\phi_2} \in M$~-- a contradiction to the uniqueness property of $M$.

\section{Conclusion and Open Questions}\label{sect:concl}
We proved that the ambiguity hierarchy is strict for regular languages over infinite trees.

For each level of the ambiguity hierarchy we provided   a language which occupies this level. It is not difficult to see
that all these languages  are  definable by MSO-formulas without  the second-order quantifiers (formulas of the first-order fragment of MSO).
Concerning the topological complexity, Olivier Finkel \cite{OF20} observed that  these languages have low topological complexity: $L_{\neg a_1 \vee \dots \vee \neg a_k} $ are closed languages;
 $L^{fa}$  and $L_{\exists a_1} $ are  countable unions of closed sets, i.e., a $\Sigma^0_2$-sets; the uncountably ambiguous language  $L_{\mathit{no-max}} $ is
    $\Pi^0_2$. On the other hand,   Skrzypczak \cite{Skrz18} proved that unambiguous languages climb up the whole index hierarchy  and  are topologically as complex as arbitrary regular  tree languages.


A natural question is whether the ambiguity degree is decidable. However, this is not a trivial matter. In \cite{DBLP:conf/csl/BilkowskiS13}
some partial solutions for variants of the problem whether a given language is unambiguous are provided.
We proved that countable regular languages are unambiguous. Since it is decidable whether a language is countable \cite{Damian},
this provides  a decidable sufficient condition for a langauge to be unambiguous. 

A less ambitious task is to develop techniques for computing degrees of ambiguity and  compute the degree of ambiguity of some natural languages.
Let $\Sigma_1:=\{c,a_1\}$ and
$L_{\exists^\infty a_1} := \{t \in T_{\Sigma_1}^\omega \mid$ there are infinitely many  $a_1$-labeled nodes in $t \}$.
$L_{\exists^\omega a_1} := \{t \in T_{\Sigma_1}^\omega \mid$ there is a branch with infinitely many  $a_1$-labeled nodes in $t \}$.
$L_{a_1-\infty\mbox{antichain} } := \{t \in T_{\Sigma_1}^\omega \mid$ the set of  $a_1$-labeled nodes in $t$ contain  an infinite    antichain$ \}$.
All  these  languages are regular. There are (Moore) reductions from $L_{\exists a_1}$ to these languages, hence they    are not  finitely ambiguous.  We believe that their ambiguity degree is uncountable, but we were unable to prove this.

We provided sufficient conditions for a language to be not finitely ambiguous and for a language to have uncountable degree of ambiguity.
In particular, we proved that the degree of ambiguity of the complement of a countable regular language is $\aleph_0$ or $2^{\aleph_0}$,
and provided natural examples of such languages with countable degree of ambiguity. We proved that the degree of ambiguity of the complement of a finite  regular language is $\aleph_0$.
Yet, it is open whether the degree of ambiguity of the complement of countable regular languages is $\aleph_0$.

%
\section*{Acknowledgments}
We would like to thank anonymous referees for their helpful suggestions.

\bibliographystyle{alpha}
\bibliography{ref}
\appendix
\section{Proof of Claim \ref{claim-regular-positional-winning-strategy}}	
\claimregularpositionalwinningstrategy*
\begin{claimproof}
	$t_0$ is regular, and therefore there is a formula $\psi_{t_0}(\sigma)$ that defines $t_0$ in the unlabeled full-binary tree.
	
        We will use $\psi_{t_0}(\sigma)$ to define the following formula $\mathit{PathfinderWins}_{\cA, t_0}(\phi, STR)$, as the conjunction of the following conditions:
	\begin{enumerate}
		\item \label{bullet-pi-conditions} $\exists \pi$ such that:
		\begin{enumerate}
			\item $\pi$ is a branch
			\item $\forall u \in \pi: (STR(u, \phi(u \cdot l), \phi(u \cdot r)) = l) \leftrightarrow u \cdot l \in \pi)$~-- the Pathfinder moves $d_0\dots d_j\dots$ are consistent with $STR$ and are along the branch $\pi$.
		\end{enumerate}
		\item $\exists \sigma: \psi_{t_0}(\sigma)$ and at least one of the following holds:
		\begin{enumerate}
			\item \label{bullet-automaton-makes-invalid-move} $\exists v \in \pi$ such that $(\phi(v), \sigma(v), \phi(v \cdot l), \phi(v \cdot r)) \notin \delta$~-- the Automaton move at $(v,\phi(v))$ is invalid.
			\item \label{bullet-max-color-odd} The maximal color that $\mathbb{C}$ assigns infinitely often to states in $\phi(\pi)$ is odd.
		\end{enumerate}
	\end{enumerate}
	
	\begin{clm} \label{claim-pathfinder-wins-formula}
          $\mathit{PathfinderWins}_{\cA, t_0}(\phi, STR)$ holds for a positional strategy $STR$ of Pathfinder and a computation $\phi$ of $\cA$ on a tree $t'$ iff the play $\overline{s}$ of $STR$ against $str_\phi$ in $G_{t_0, \cA}$ is winning for Pathfinder.
	\end{clm}
	\begin{claimproof}
		By definition of $G_{t_0, \cA}$, Pathfinder wins if either Automaton makes an invalid move (condition \ref{bullet-automaton-makes-invalid-move}) or the maximal color that is assigned infinitely often to the positions in $\pi_{\overline{s}}$ is odd. Since all Pathfinder positions have color $0$, this is equivalent to the maximal color assigned infinitely often to Automaton positions being odd.
		
		Let $\overline{s} = e_0, d_0, e_1, d_1, \dots, e_i, d_i, \dots$. Notice that by condition \ref{bullet-pi-conditions}, there is a unique branch~$\pi$ such that $\pi = v_0, \dots v_i, \dots$ where $v_i = d_0 \dots d_{i-1}$.
		By Claim \ref{claim-computation-to-strategy}, we have $\phi(v_i) = q_i$, where the $i$-th position of Automaton in $\pi_{\overline{s}}$ is $(v_i, q_i)$. Since $\mathbb{C}_G(v_i, q_i) = \mathbb{C}(q_i)$, we conclude that the maximal color that $\mathbb{C}$ assigns infinitely often to states in $\phi(\pi)$ is odd iff the maximal color that $\mathbb{C}_G$ assigns infinitely often to positions in $\pi_{\overline{s}}$ is odd. This is assured by condition~\ref{bullet-max-color-odd}.
	\end{claimproof}

        \noindent Let $\mathit{WinningStrategy}_{t_0, \cA}(STR) := \forall \phi$ such that the following holds:
	\begin{itemize}
		\item If there is $t$ such that $\phi$ is an accepting computation of $\cA$ on $t$, then:
		\begin{itemize}
                  \item $\mathit{PathfinderWins}_{\cA, t_0}(\phi, STR)$ holds
		\end{itemize}
	\end{itemize}
        Recalling that the set of all computation of $\cA$ is MSO-definable, we can conclude that $\mathit{WinningStrategy}_{t_0, \cA}(STR)$ is MSO-definable in the unlabeled full-binary tree.
	\begin{clm} \label{claim-winningstrategy-formula}
          $\mathit{WinningStrategy}_{t_0, \cA}(STR)$ holds for a positional strategy $STR$ of Pathfinder iff $STR$ is a positional winning strategy of Pathfinder.
	\end{clm}
	\begin{claimproof}
		$\Rightarrow$: By Claim \ref{claim-pathfinder-wins-formula}, $STR$ wins in $G_{t_0, \cA}$ against each positional strategy of Automaton. Assume, for the sake of contradiction, that is a non-positional strategy $str'$ of automaton that wins against $STR$. Then by positional determinacy of parity games, we conclude that there is a positional strategy $str''$ that wins against $STR$~-- a contradiction.
		
		$\Leftarrow$: Follows immediately from Claim \ref{claim-pathfinder-wins-formula}.
	\end{claimproof}
        We have $t_0 \notin L(\cA)$ and therefore by Claim \ref{claim-computation-to-strategy}(3), Automaton does not have a positional winning strategy. From positional determinacy of parity games we conclude that Pathfinder has a positional winning strategy. Therefore, there is a strategy $STR'$ that satisfies $\mathit{WinningStrategy}_{t_0, \cA}(STR)$ in the unlabeled full-binary tree.
	
        Therefore, $\mathit{WinningStrategy}_{t_0, \cA}(STR)$ defines a non-empty tree language over alphabet $Q\times Q \rightarrow \{l,r\}$. By Rabin's basis Theorem, we conclude that there is a regular tree $\widehat{STR}$ in this language, and by Claim \ref{claim-winningstrategy-formula} we conclude that $\widehat{STR}$ is a positional winning strategy for Pathfinder in $G_{t_0, \cA}$.
\end{claimproof}
\paragraph{Remark
(Logic Free Proof of Claim \ref{claim-regular-positional-winning-strategy})}
One can reduce a membership game for a regular tree $t_0$ to a game on a finite graph. By positional determinacy Theorem, Pathfinder will
have a positional winning strategy in the reduced game. From this strategy a regular winning strategy in $G_{t_0, \cA}$ for Pathfinder is easily constructed.
\end{document}